\documentclass[11pt]{article}

\usepackage{times}

\usepackage[paperwidth=199.8mm,paperheight=297mm,centering,hmargin=25mm,vmargin=25mm]{geometry}
\usepackage{authblk} % package for author list
\usepackage[titletoc,title]{appendix}
\usepackage{hyperref}
\hypersetup{colorlinks=true,citecolor=blue,linkcolor=blue,filecolor=blue,urlcolor=blue,breaklinks=true}
\usepackage[dvipsnames]{xcolor}
\usepackage{color}
\usepackage[marginal]{footmisc}
\usepackage{url}
\usepackage{mathtools}
\usepackage{amsmath}
\usepackage{amssymb}
\usepackage[shortlabels]{enumitem}
\usepackage{graphicx,epic,eepic,epsfig,latexsym,verbatim}
\usepackage{dsfont}

\usepackage{framed}
\definecolor{shadecolor}{rgb}{0.9,0.9,0.9}
\usepackage{subcaption}
\usepackage{caption}
\usepackage{float}
\usepackage{tikz}
\usepackage{tcolorbox}
\usepackage{relsize}

\definecolor{LightGray}{RGB}{220,220,220}
\definecolor{myred}{RGB}{236, 17, 0}
\definecolor{myblue}{RGB}{10, 88, 153}
\definecolor{myorange}{RGB}{236, 137, 0}
\definecolor{mygreen}{RGB}{26, 152, 81}

\usepackage{amsthm}
\newtheorem{definition}{Definition}
\newtheorem{proposition}[definition]{Proposition}
\newtheorem{lemma}[definition]{Lemma}
\newtheorem{theorem}[definition]{Theorem}
\newtheorem{corollary}[definition]{Corollary}
\newtheorem{assumption}[definition]{Assumption}
\def\squareforqed{\hbox{\rlap{$\sqcap$}$\sqcup$}}
\def\qed{\ifmmode\squareforqed\else{\unskip\nobreak\hfil
\penalty50\hskip1em\null\nobreak\hfil\squareforqed
\parfillskip=0pt\finalhyphendemerits=0\endgraf}\fi}
\def\endenv{\ifmmode\;\else{\unskip\nobreak\hfil
\penalty50\hskip1em\null\nobreak\hfil\;
\parfillskip=0pt\finalhyphendemerits=0\endgraf}\fi}

\newcounter{remark}
\newenvironment{remark}[1][]{\refstepcounter{remark}\par\medskip\noindent%
\textbf{Remark~\theremark #1} }{\medskip}

\newcounter{example}
\newenvironment{example}[1][]{\refstepcounter{example}\par\medskip\noindent%
\textbf{Example~\theexample #1} }{\medskip}

\mathchardef\ordinarycolon\mathcode`\:
\mathcode`\:=\string"8000
\def\vcentcolon{\mathrel{\mathop\ordinarycolon}}
\begingroup \catcode`\:=\active
  \lowercase{\endgroup
  \let :\vcentcolon
  }

\usepackage{colortbl}
\usepackage{pifont}% http://ctan.org/pkg/pifont
\definecolor{Gray}{gray}{0.92}
\definecolor{Gray2}{gray}{0.75}
\definecolor{maroon}{cmyk}{0,0.87,0.68,0.32}
\usepackage{booktabs}
\usepackage{makecell}%To keep spacing of text in tables
\usepackage{diagbox}
\usepackage{multirow}

\newcommand{\nc}{\newcommand}
\nc{\rnc}{\renewcommand}
\nc{\bra}[1]{\langle#1|}
\nc{\ket}[1]{|#1\rangle}
\nc{\<}{\langle}
\rnc{\>}{\rangle}
\nc{\ketbra}[2]{|#1\rangle\!\langle#2|}
\nc{\braket}[2]{\langle#1|#2\rangle}
\nc{\braandket}[3]{\langle #1|#2|#3\rangle}
\nc{\proj}[1]{| #1\rangle\!\langle #1 |}
\nc{\avg}[1]{\langle#1\rangle}
\nc{\Rank}{\operatorname{Rank}}
\nc{\smfrac}[2]{\mbox{$\frac{#1}{#2}$}}
\nc{\tr}{\operatorname{Tr}}
\nc{\ox}{\otimes}

\nc{\cA}{{\cal A}}
\nc{\cB}{{\cal B}}
\nc{\cC}{{\cal C}}
\nc{\cD}{{\cal D}}
\nc{\cE}{{\cal E}}
\nc{\cF}{{\cal F}}
\nc{\cG}{{\cal G}}
\nc{\cH}{{\cal H}}
\nc{\cI}{{\cal I}}
\nc{\cJ}{{\cal J}}
\nc{\cK}{{\cal K}}
\nc{\cL}{{\cal L}}
\nc{\cM}{{\cal M}}
\nc{\cN}{{\cal N}}
\nc{\cO}{{\cal O}}
\nc{\cP}{{\cal P}}
\nc{\cQ}{{\cal Q}}
\nc{\cR}{{\cal R}}
\nc{\cS}{{\cal S}}
\nc{\cT}{{\cal T}}
\nc{\cV}{{\cal V}}
\nc{\cU}{{\cal U}}
\nc{\cX}{{\cal X}}
\nc{\cY}{{\cal Y}}
\nc{\cZ}{{\cal Z}}
\nc{\cW}{{\cal W}}

\nc{\RR}{{{\mathbb R}}}
\nc{\CC}{{{\mathbb C}}}
\nc{\FF}{{{\mathbb F}}}
\nc{\NN}{{{\mathbb N}}}
\nc{\ZZ}{{{\mathbb Z}}}
\nc{\QQ}{{{\mathbb Q}}}
\nc{\UU}{{{\mathbb U}}}
\nc{\EE}{{{\mathbb E}}}
\nc{\id}{{\operatorname{id}}}

\nc{\1}{{\mathds{1}}}
\nc{\supp}{{\operatorname{supp}}}

\def\ve{\varepsilon}

\makeatletter
\def\grd@save@target#1{%
  \def\grd@target{#1}}
\def\grd@save@start#1{%
  \def\grd@start{#1}}
\tikzset{
  grid with coordinates/.style={
    to path={%
      \pgfextra{%
        \edef\grd@@target{(\tikztotarget)}%
        \tikz@scan@one@point\grd@save@target\grd@@target\relax
        \edef\grd@@start{(\tikztostart)}%
        \tikz@scan@one@point\grd@save@start\grd@@start\relax
        \draw[minor help lines,magenta] (\tikztostart) grid (\tikztotarget);
        \draw[major help lines] (\tikztostart) grid (\tikztotarget);
        \grd@start
        \pgfmathsetmacro{\grd@xa}{\the\pgf@x/1cm}
        \pgfmathsetmacro{\grd@ya}{\the\pgf@y/1cm}
        \grd@target
        \pgfmathsetmacro{\grd@xb}{\the\pgf@x/1cm}
        \pgfmathsetmacro{\grd@yb}{\the\pgf@y/1cm}
        \pgfmathsetmacro{\grd@xc}{\grd@xa + \pgfkeysvalueof{/tikz/grid with coordinates/major step}}
        \pgfmathsetmacro{\grd@yc}{\grd@ya + \pgfkeysvalueof{/tikz/grid with coordinates/major step}}
        \foreach \x in {\grd@xa,\grd@xc,...,\grd@xb}
        \node[anchor=north] at (\x,\grd@ya) {\pgfmathprintnumber{\x}};
        \foreach \y in {\grd@ya,\grd@yc,...,\grd@yb}
        \node[anchor=east] at (\grd@xa,\y) {\pgfmathprintnumber{\y}};
      }
    }
  },
  minor help lines/.style={
    help lines,
    step=\pgfkeysvalueof{/tikz/grid with coordinates/minor step}
  },
  major help lines/.style={
    help lines,
    line width=\pgfkeysvalueof{/tikz/grid with coordinates/major line width},
    step=\pgfkeysvalueof{/tikz/grid with coordinates/major step}
  },
  grid with coordinates/.cd,
  minor step/.initial=.2,
  major step/.initial=1,
  major line width/.initial=2pt,
}
\makeatother

\makeatletter
\newcommand*\rel@kern[1]{\kern#1\dimexpr\macc@kerna}
\newcommand*\widebar[1]{%
  \begingroup
  \def\mathaccent##1##2{%
    \rel@kern{0.8}%
    \overline{\rel@kern{-0.8}\macc@nucleus\rel@kern{0.2}}%
    \rel@kern{-0.2}%
  }%
  \macc@depth\@ne
  \let\math@bgroup\@empty \let\math@egroup\macc@set@skewchar
  \mathsurround\z@ \frozen@everymath{\mathgroup\macc@group\relax}%
  \macc@set@skewchar\relax
  \let\mathaccentV\macc@nested@a
  \macc@nested@a\relax111{#1}%
  \endgroup
}
\makeatother

\newcommand{\reg}{\infty}

\newcommand{\Renyi}{R\'{e}nyi }

\usepackage{mathrsfs}
\nc{\pl}{{\scalebox{0.7}{+}}}
\nc{\PSD}{\HERM_{\pl}}
\nc{\PD}{\HERM_{\pl\pl}}
\nc{\polarPSD}[1]{{#1}_{\pl}^{\circ}}
\nc{\polarPSDre}[1]{{#1}_{\pl}^{\star}}
\nc{\polarPD}[1]{{#1}_{\pl\pl}^{\circ}}
\nc{\HERM}{\mathscr{H}}
\nc{\cvxset}{\mathscr{C}}
\nc{\density}{\mathscr{D}}
\nc{\subdensity}{\mathscr{D}_\bullet}
\nc{\bcC}{\cC}
\nc{\Meas}{{\scriptscriptstyle \rm M}}
\nc{\Proj}{{{\scriptscriptstyle \rm P}}}
\nc{\RM}{{{\mathscr{R}}}}
\nc{\sK}{{{\mathscr{K}}}}
\nc{\sS}{{{\mathscr{S}}}}
\nc{\sT}{{{\mathscr{T}}}}
\nc{\sA}{{{\mathscr{A}}}}
\nc{\sB}{{{\mathscr{B}}}}
\nc{\sC}{{{\mathscr{C}}}}
\nc{\sE}{{{\mathscr{E}}}}
\nc{\sL}{{{\mathscr{L}}}}
\nc{\sG}{{{\mathscr{G}}}}
\nc{\sF}{{{\mathscr{F}}}}
\nc{\sI}{{{\mathscr{I}}}}
\nc{\sN}{{{\mathscr{N}}}}
\nc{\sM}{{{\mathscr{M}}}}
\nc{\END}{\operatorname{End}}
\nc{\PERM}{\mathfrak{\sigma}}

\nc{\Cone}{\text{\rm Cone}}
\nc{\sep}{{\SEP}}
\nc{\BS}{{\scriptscriptstyle \rm {BS}}}
\nc{\Sand}{{\scriptscriptstyle  \rm S}}
\nc{\Petz}{{\scriptscriptstyle  \rm P}}
\nc{\Hypo}{{\scriptscriptstyle  \rm H}}
\nc{\DD}{{{\mathbb D}}}
\nc{\suchthat}{\text{\rm s.t.}}
\nc{\PPT}{\text{\rm PPT}}
\nc{\Rains}{\text{\rm Rains}}
\nc{\WD}{\text{\rm WD}}
\nc{\new}{\text{\rm new}}
\nc{\sfT}{\mathsf T}
\nc{\SEP}{\text{\rm SEP}}
\nc{\PSEP}{\text{\rm PSEP}}
\nc{\CPTP}{\text{\rm CPTP}}
\nc{\POVM}{\text{\rm POVM}}
\nc{\PVM}{\text{\rm PVM}}
\nc{\CP}{\text{\rm CP}}
\nc{\adv}{\text{\rm adv}}
\nc{\spec}{\text{\rm spec}}
\nc{\poly}{\text{\rm poly}}
\nc{\End}{\operatorname{End}}
\nc{\Par}{\operatorname{Par}}
\nc{\RNG}{\operatorname{RNG}}
\nc{\epi}{\boldsymbol{\operatorname{epi}}}
\nc{\op}{\boldsymbol{\operatorname{op}}}

\usepackage{stmaryrd}
\nc{\db}[1]{\left\llbracket#1 \right\rrbracket}

\usepackage{pgfplots}

\nc{\img}{\mathbf{i}}

\usepackage{changepage}

\begin{document}

\title{\Large \textbf{Efficient approximation of regularized relative entropies\\ and applications}}

\author[1]{Kun Fang \thanks{kunfang@cuhk.edu.cn}}
\author[2]{Hamza Fawzi \thanks{h.fawzi@damtp.cam.ac.uk}}
\author[3]{Omar Fawzi \thanks{omar.fawzi@ens-lyon.fr}}

\affil[1]{\small School of Data Science, The Chinese University of Hong Kong, Shenzhen,\protect\\  Guangdong, 518172, China}
\affil[2]{\small Department of Applied Mathematics and Theoretical Physics,  University of Cambridge, \protect\\ Cambridge CB3 0WA, United Kingdom}
\affil[3]{\small Univ Lyon, Inria, ENS Lyon, UCBL, LIP, Lyon, France}

\date{}

\maketitle

\begin{abstract}
The quantum relative entropy is a fundamental quantity in quantum information science, characterizing the distinguishability between two quantum states. However, this quantity is not additive in general for correlated quantum states, necessitating regularization for precise characterization of the operational tasks of interest. Recently, we proposed the study of the regularized relative entropy between two sequences of sets of quantum states in [arXiv: 2411.04035], which captures a general framework for a wide range of quantum information tasks. Here, we show that given suitable structural assumptions and efficient descriptions of the sets, the regularized relative entropy can be efficiently approximated within an additive error by a quantum relative entropy program of polynomial size. This applies in particular to the regularized relative entropy in adversarial quantum channel discrimination. Moreover, we apply the idea of efficient approximation to quantum resource theories. In particular, when the set of interest does not directly satisfy the required structural assumptions, it can be relaxed to one that does. This provides improved and efficient bounds for the entanglement cost of quantum states and channels, entanglement distillation and magic state distillation. Numerical results demonstrate improvements even for the first level of approximation.
\end{abstract}

% \newpage

{
\hypersetup{linkcolor=black}
\tableofcontents
}

% \newpage
\section{Introduction}

The classical relative entropy, also known as the Kullback-Leibler (KL) divergence~\cite{kullback1951information}, is a measure of how much a model probability distribution is different from a true probability distribution. It plays a pivotal role in classical information processing and finds applications in diverse domains including machine learning~\cite{murphy2012machine}, data compression~\cite{cover1999elements,gray2011entropy} and statistical mechanics~\cite{jaynes1957information,crooks1999entropy}. With the development of quantum information science, the quantum relative entropy has been proposed as a quantum generalization of the KL divergence~\cite{umegaki1962conditional},
\begin{align}
    D(\rho\|\sigma) := \tr [\rho (\log \rho - \log \sigma)]
\end{align}
quantifying the distinguishability between quantum states $\rho$ and $\sigma$~\cite{hiai1991proper}. It has found widespread applications in various fields~\cite{schumacher2002relative,vedral2002role}, including quantum machine learning~\cite{biamonte2017quantum}, quantum channel coding~\cite{hayashi2017quantum}, quantum error correction~\cite{cerf1997information}, quantum resource theories~\cite{chitambar2019quantum}, and quantum cryptography~\cite{pirandola2020advances}.

A useful property for quantum relative entropy is the additivity between two tensor product states, i.e., $D(\rho_1\ox \rho_2\|\sigma_1\ox \sigma_2) = D(\rho_1\|\sigma_1) + D(\rho_2\|\sigma_2)$. However, this property does not hold for general correlated quantum states, i.e., $D(\rho_{12}\|\sigma_{12}) \neq D(\rho_1\|\sigma_1) + D(\rho_2\|\sigma_2)$, necessitating regularization for precise characterization of operational tasks of interest. Recently, we proposed in~\cite{fang2024generalized} the study of the regularized quantum relative entropy between two sequences of sets of quantum states $\sA_n$, $\sB_n$ acting on $\cH^{\otimes n}$ for some Hilbert space $\cH$: 
\begin{align}
    D^{\infty}(\sA\|\sB) := \lim_{n\to \infty}\frac{1}{n}D(\sA_n\|\sB_n),
\end{align} 
where $D(\sA_n\|\sB_n) = \inf_{\rho \in \sA_n, \sigma \in \sB_n} D(\rho \| \sigma)$. This quantity captures a general framework for a wide range of quantum information tasks. This includes quantum hypothesis testing~\cite{hiai1991proper,Ogawa2000,Brand_o_2010,brandao2020adversarial,berta2021composite,Mosonyi_2022,hayashi2024generalized,lami2024solutiongeneralisedquantumsteins} and quantum channel discrimination \cite{brandao2020adversarial,fang2020chain,fang2021ultimate,fang2025adversarial} from a foundational perspective, entropy accumulation theorems for quantum cryptography~\cite{dupuis2020entropy,metger2022generalised}, and quantum resource distillation \cite{Brand_o_2010,rains2001semidefinite,fang2019non,regula2019one,Bravyi_2005,veitch2014resource,fang2020no,fang2022no,Wang2018magicstates,hayashi2021finite,diaz2018using} and preparation \cite{Wang2017,Wang2017irreversibility,wang2023computable,lami2023no}, which are crucial for quantum computing and quantum networking~\cite{fang2023quantum}.

In general, computing $D^{\infty}(\sA\|\sB)$ is challenging due to the limit, which we refer to as ``regularization''. One notable example is given by the regularized relative entropy of entanglement 
\begin{align}
    D^\reg(\rho_{AB}\|\SEP):= \lim_{n\to \infty} \frac{1}{n} D(\rho_{AB}^{\ox n}\|\SEP(A^{n}:B^n)),
\end{align}
where $\SEP(A^n:B^n)$ denotes the set of all separable states between Hilbert spaces $\cH_A^{\ox n}$ and $\cH_B^{\ox n}$. This quantity uniquely determines the ultimate limits of entanglement manipulation and serves as the key quantity in understanding the second law of quantum entanglement~\cite{Brand_o_2010,hayashi2024generalized,lami2024solutiongeneralisedquantumsteins}. However, evaluating this quantity is extremly hard in general, as it involves the regularization as well as the separability problem~\cite{gurvits2003classical}. 

\paragraph{Our results} In this work, we show in Theorem~\ref{thm: efficient relative entropy program} that given suitable structural assumptions (in particular the stability of the polar set under tensor product) and efficient descriptions of the sets, the regularized relative entropy $D^{\infty}(\sA\|\sB)$ can be efficiently approximated within an additive error by a quantum relative entropy program of polynomial size. We then apply this result in Section~\ref{sec: applications} to several problems in quantum information theory. The first application (Section~\ref{sec: Adversarial quantum channel discrimination}) is to compute the regularized relative entropy between the image sets of two quantum channels, which characterizes the optimal exponent in adversarial channel discrimination~\cite{fang2025adversarial}. For this problem, the relevant sets satisfy the structural assumptions. In the following applications, this will not be the case, but the sets can be relaxed to sets that do satisfy the requirements.
For instance, in entanglement theory, the set of separable states can be relaxed to the Rains set~\cite{rains2001semidefinite,audenaert2002asymptotic}, which satisfies all necessary assumptions. We illustrate this by obtaining bounds (Section~\ref{sec: Entanglement cost for quantum states and channels}) on the entanglement cost of quantum states and channels improving on~\cite{Wang2017,Wang2017irreversibility,wang2023computable,lami2023no}.  Numerical results demonstrate improvements even for the first level of approximation. This approach can also be applied to obtain improved bounds on entanglement distillation~\cite{rains2001semidefinite,audenaert2002asymptotic,Brand_o_2010} as discussed in Section~\ref{sec: Quantum entanglement distillation}. Similarly, in fault-tolerant quantum computing, the set of stabilizer states can be relaxed to the set of states with non-positive mana~\cite{veitch2014resource}, which also fulfills the required conditions. As described in Section~\ref{sec: Magic state distillation}, this can be used to obtain improved bounds for magic state distillation.

Generally, our result can be applied by verifying the conditions of the relevant theory and performing necessary relaxations when required. Therefore, we anticipate that this approach has the potential for other applications beyond the specific cases discussed here.

\section{Preliminaries}
\label{sec: preliminaries}

\subsection{Notations}

In this section we set the notations and define several quantities that will be used throughout this work. Some frequently used notations are summarized in Table~\ref{tab: state version}. Note that we label different physical systems by capital Latin letters and use these labels as subscripts to guide the reader by indicating which system a mathematical object belongs to. We drop the subscripts when they are evident in the context of an expression (or if we are not talking about a specific system).

\setlength\extrarowheight{2pt}
\begin{table}[H]
% \makegapedcells
\centering
\begin{tabular}{l|l}
\toprule[2pt]
Notations & Descriptions\\
\hline
$\cH_A$ & Hilbert space on system $A$\\
% $|A|$ & Dimension of $\cH_A$\\
$\sL(A)$ & Linear operators on $\cH_A$\\
$\HERM(A)$ & Hermitian operators on $\cH_A$\\
$\HERM_{\pl}(A)$ & Positive semidefinite operators on $\cH_A$\\
$\HERM_{\pl\pl}(A)$ & Positive definite operators on $\cH_A$\\
$\density(A)$ & Density matrices on $\cH_A$\\
$\sA,\sB,\sC$ & Set of linear operators\\
$\cvxset^{\circ}$ & Polar set $\cvxset^\circ := \{X: \tr[XY] \leq 1,  \forall\, Y\in \cvxset\}$ of $\cvxset$\\
$\polarPSD{\cvxset}$ & Polar set restricted to positive semidefinite cone $\cvxset^\circ \cap \PSD$\\
$\polarPD{\cvxset}$ & Polar set restricted to positive definite operators $\cvxset^\circ \cap \PD$\\
$\CPTP$ & Completely positive and trace preserving maps\\
$\CP$ & Completely positive maps\\
$\log(x)$ & Logarithm of $x$ in base two\\
\bottomrule[2pt]  
\end{tabular}
\caption{\small Overview of notational conventions.}
\label{tab: state version}
\end{table}

\subsection{Polar set and support function}

In the following, we introduce the definitions of the polar set and support function, along with a fundamental result that will be used in our discussions.

\begin{definition}
Let $\cvxset \subseteq \sE$ be a convex set in some Euclidean space $\sE$. Its polar set is defined by
\begin{align}
\cvxset^\circ:= \{X \in \sE: \tr[XY] \leq 1,  \forall\, Y\in \cvxset\}.
\end{align} 
Let $\polarPSD{\cvxset}:= \cvxset^\circ \cap \PSD$ and  $\polarPD{\cvxset}:= \cvxset^\circ \cap \PD$ be the intersections with positive semidefinite operators and positive definite operators, respectively.
The support function of $\cvxset$ at $\omega$ is defined by $h_{\cvxset}(\omega)= \sup_{\sigma \in \cvxset} \tr[\sigma \omega]$. 
\end{definition}

It is clear from the definitions that $\omega \in {\cvxset^\circ}$ if and only if $h_{\cvxset}(\omega) \leq 1$.

\begin{definition}\label{def: closed under tensor product}
    Let $\cH_1$ and $\cH_2$ be finite-dimensional Hilbert spaces. Consider three sets $\sA_1 \subseteq \PSD(\cH_1)$, $\sA_2 \subseteq \PSD(\cH_2)$, and $\sA_{12} \subseteq \PSD(\cH_1 \otimes \cH_2)$. We call $\{\sA_1,\sA_2,\sA_{12}\}$ is {closed under tensor product} if for any $X_1 \in \sA_1$, $X_2 \in \sA_2$, we have $X_1 \otimes X_2 \in \sA_{12}$. In short, we write $\sA_1 \ox \sA_2 \subseteq \sA_{12}$.
\end{definition}

The following lemma provides an equivalent condition for determining if the polar sets of interest are closed under tensor product, which can be easier to validate for specific examples.

\begin{lemma}\cite[Lemma 7]{fang2024generalized}\label{lema: polar set and support function}
Let $\cH_1$ and $\cH_2$ be finite-dimensional Hilbert spaces. Consider three sets $\sA_1 \subseteq \PSD(\cH_1)$, $\sA_2 \subseteq \PSD(\cH_2)$, and $\sA_{12} \subseteq \PSD(\cH_1 \otimes \cH_2)$. Their polar sets are closed under tensor product if and only if their support functions are sub-multiplicative. That is,
\begin{align}\label{eq: polar set and support function tmp1}
\polarPSD{(\sA_{1})} \ox \polarPSD{(\sA_{2})} \subseteq \polarPSD{(\sA_{12})} \iff h_{\sA_{12}}(X_1 \ox X_2) \leq h_{\sA_1}(X_1) h_{\sA_2}(X_2),\; \forall X_i \in \PSD(\cH_i).
\end{align}
\end{lemma}

\subsection{Quantum divergences}

A functional $\DD: \density \times \PSD \to \RR$ is a quantum divergence if it satisfies the data-processing inequality $\DD(\cE(\rho)\|\cE(\sigma)) \leq \DD(\rho\|\sigma)$ for any CPTP map $\cE$. In the following, we will introduce several quantum divergences and their fundamental properties, which will be used throughout this work. Additionally, we will define quantum divergences between two sets of quantum states, which will be the main quantity of interest in this work.

\begin{definition}(Umegaki relative entropy~\cite{umegaki1954conditional}.)
For any $\rho\in \density$ and $\sigma \in \PSD$, the Umegaki relative entropy is defined by
\begin{align}\label{eq: Umegaki}
    D(\rho\|\sigma):= \tr[\rho(\log \rho - \log \sigma)],
\end{align}
if $\supp(\rho) \subseteq \supp(\sigma)$ and $+\infty$ otherwise.
\end{definition}

The min-relative entropy is defined by
\begin{align}
    D_{\min}(\rho\|\sigma) := -\log \tr [\Pi_\rho \sigma],
\end{align}
with $\Pi_\rho$ the projection on the support of $\rho$. The max-relative entropy is defined by~\cite{datta2009min,renner2005security},
\begin{align}\label{eq: definition of Dmax}
D_{\max}(\rho\|\sigma):= \log\inf\big\{t \in \RR \;:\; \rho \leq t\sigma \big\}\;,
\end{align}
if $\supp(\rho) \subseteq \supp(\sigma)$ and $+\infty$ otherwise.

\begin{definition}(Measured relative entropy~\cite{donald1986relative,hiai1991proper}.)
For any $\rho \in \density$ and $\sigma \in \PSD$, the measured relative entropy is defined by
\begin{align}
D_{\Meas} (\rho\|\sigma) := \sup_{(\cX,M)} D(P_{\rho,M}\|P_{\sigma,M}),
\end{align}
where $D$ is the Kullback–Leibler divergence and the optimization is over finite sets $\cX$ and positive operator valued measures $M$ on $\cX$ such that $M_x \geq 0$ and $\sum_{x \in \cX} M_x = I$, $P_{\rho,M}$ is a measure on $\cX$ defined via the relation $P_{\rho,M}(x) = \tr[M_x\rho]$ for any $x \in \cX$.    
\end{definition}

A variational expression for $D_{\Meas}$ is given by~\cite[Lemma 1]{Berta2017},
\begin{align}\label{eq: DM variational}
    D_{\Meas}(\rho\|\sigma) = \sup_{\omega \in \PD} \ \tr[\rho \log \omega] + 1  - \tr[\sigma \omega].
\end{align}

\begin{definition}(Measured \Renyi divergence~\cite{Berta2017}.)
    Let $\alpha \in (0,1) \cup (1,\infty)$.  For any $\rho \in \density$ and $\sigma \in \PSD$, the {measured \Renyi divergence} is defined as
    \begin{align}\label{eq: definition DM alpha}
    D_{\Meas, \alpha} (\rho\|\sigma) := \sup_{(\cX,M)} D_{\alpha}(P_{\rho,M}\|P_{\sigma,M}),
    \end{align}
    where $D_{\alpha}$ is the classical \Renyi divergence. 
\end{definition}

The following result shows the ordering relation among different relative entropies.

% \begin{shaded}
\begin{lemma}\label{thm: comparison of quantum divergence}
Let $\alpha \in [1/2,1)$. For any $\rho \in \density$ and $\sigma \in \PSD$, 
\begin{align}
	D_{\min}(\rho\|\sigma) \leq D_{\Meas, \alpha}(\rho\|\sigma) \leq D_{\Meas}(\rho\|\sigma) \leq  D(\rho\|\sigma).
\end{align}
\end{lemma}
% \end{shaded}
\begin{proof}
The last two inequalities follow from the monotonicity in $\alpha$ of the classical R\'enyi divergences and the data processing inequality for $D$. As $D_{\Meas, \alpha}$ is monotone increasing in $\alpha$, it remains to show the first inequality for $\alpha = \frac12$. By the variational formula in~\cite[Eq. (21)]{Berta2017}, we have
\begin{align}
D_{\Meas, 1/2}(\rho\|\sigma) = -\log \inf_{\omega \in \PD} \tr[\rho \omega^{-1}] \tr[\sigma \omega],
\end{align}
which is the same as the Alberti's theorem for quantum fidelity~(see e.g.~\cite[Corollary 3.20]{watrous2018theory}). Consider a feasible solution $\omega_\ve = \Pi_\rho + \ve(I - \Pi_\rho) \in \PD$ with $\ve > 0$. It gives $D_{\Meas, 1/2}(\rho\|\sigma) \geq -\log \tr[\rho \omega_\ve^{-1}] \tr[\sigma \omega_\ve]$. Since $\rho$ has trace one, it gives $\tr[\rho \omega_\ve^{-1}] = \tr[\rho] = 1$. Then we have \begin{align}
D_{\Meas, 1/2}(\rho\|\sigma) \geq -\log \tr[\sigma \omega_\ve] = -\log [(1-\ve) \tr \Pi_\rho \sigma + \ve].
\end{align} 
As the above holds for any $\ve > 0$, we take $\ve \to 0^{+}$ and get $D_{\Meas, 1/2}(\rho\|\sigma) \geq -\log \tr[\Pi_\rho \sigma] = D_{\min}(\rho\|\sigma)$, which completes the proof.
\end{proof}

In this work, we will focus on the study of quantum divergences between two sets of quantum states.

\begin{definition}(Quantum divergence between two sets of states.)\label{def: divergence between two sets}
    Let $\DD$ be a quantum divergence between two quantum states. Then for any sets $\sA\subseteq \density$ and $\sB\subseteq \PSD$, the quantum divergence between these two sets of quantum states is defined by
    \begin{align}
        \DD(\sA\|\sB):= \inf_{\substack{\rho \in \sA\\ \sigma \in \sB}} \DD(\rho\|\sigma). 
    \end{align}
\end{definition}
From the geometric perspective, this quantity characterizes the distance between two sets $\sA$ and $\sB$ under the ``distance metric'' $\DD$. In particular, if $\sA = \{\rho\}$ is a singleton, we write $\DD(\rho\|\sB):= \DD(\{\rho\}\|\sB)$. For two sequences of sets $\{\sA_n\}_{n\in \NN}$ and $\{\sB_n\}_{n \in \NN}$, the regularized divergence is defined by 
\begin{align}
\DD^{\reg}(\sA \| \sB) := \lim_{n \to \infty} \frac{1}{n} \DD(\sA_{n} \| \sB_{n}),
\end{align}
whenever the limit on the right-hand side exists. 

In particular, we will focus on the sequences of sets satisfying the following assumptions.

\begin{assumption}\label{ass: steins lemma assumptions}
    Consider a family of sets $\{\sA_n\}_{n\in \NN}$ satisfying the following properties,
    \begin{itemize}
        \item (A.1) Each $\sA_n$ is convex and compact;
        \item (A.2) Each $\sA_n$ is permutation-invariant; 
        \item (A.3) $\sA_m \ox \sA_k \subseteq \sA_{m+k}$, for all $m,k \in \NN$;
        \item (A.4) $\polarPSD{(\sA_m)} \ox \polarPSD{(\sA_k)} \subseteq \polarPSD{(\sA_{m+k})}$, for all $m,k \in \NN$.
    \end{itemize}
\end{assumption}

\section{Efficient approximation of regularized relative entropies}
\label{Efficient approximation of regularized relative entropies}

In this section, we demonstrate that the regularized relative entropy $D^\reg(\sA\|\sB)$ can be efficiently approximated given efficient descriptions of $\{\sA_n\}_{n\in \NN}$ and $\{\sB_n\}_{n\in \NN}$. The main result is presented below, with some technical definitions provided later.

\begin{shaded}
    \begin{theorem}[Efficient approximation of regularized relative entropies]\label{thm: efficient relative entropy program}
    Let $\cH$ be a Hilbert space of finite dimension $d$. Let $\{\sA_n\}_{n\in\NN}$ and $\{\sB_n\}_{n\in\NN}$ be two sequences of sets satisfying Assumption~\ref{ass: steins lemma assumptions} and $\sA_n \subseteq \density(\cH^{\ox n})$, $\sB_n \subseteq \PSD(\cH^{\ox n})$ and $D_{\max}(\sA_n\|\sB_n) \leq cn$, for all $n \in \NN$ and a constant $c \in \RR_{\pl}$. If each $\sA_m$ and $\sB_m$ have semidefinite (SDP) representations of size $s^{2m}$ satisfying the symmetry conditions in Lemma~\ref{lem: permutation invariance restriction 1}, then both $D(\sA_m \| \sB_m)$ and $D_{\Meas}(\sA_m \| \sB_m)$ can  be computed by quantum relative entropy programs of size $O((m+1)^{k^2})$ with $k=\max\{s,d\}$. 
    
    As a result, $D^{\reg}(\sA \| \sB)$ can be approximated within additive error $\delta$ by a quantum relative entropy program of size $O((m_0+1)^{k^2})$ with $m_0 = \lceil \frac{8d^2}{\delta}\log \frac{d^2}{\delta}\rceil$.
    \end{theorem}
\end{shaded}

Note that quantum relative entropy programs can be efficiently solved using interior point methods, as shown in~\cite{fawzi2023optimal,he2024exploiting}. Moreover, a numerical toolkit named QICS (Quantum Information Conic Solver) has been developed in~\cite{he2024qics} to solve these programs efficiently in practice.

The remaining part of this section will provide a proof of the above result. We begin by defining quantum relative entropy programs and then show that the converging bounds for $D^\reg(\sA\|\sB)$, derived from~\cite{fang2024generalized} under structural assumptions on the sets $\sA_n$ and $\sB_n$, allow for an approximation of $D^\reg(\sA\|\sB)$ using quantum relative entropy programs. We then exploit the symmetry in these sets to reduce the size of the convex programs.

\subsection{Quantum relative entropy programs}

\begin{definition}[Conic program]
A conic program over a convex cone $K \subseteq \sE$ is an optimization problem of the form $\min_{x \in \sE} \{ \<c,x\> : x \in K,\, F(x) = g\}$ where $F:\sE \to \sF$ is a linear map and $g \in \sF$~\cite[Chapter 2]{ben2001lectures}.
The size of the program is defined as $\dim(K)$.
\end{definition}

Motivated by the conic program, we introduce the conic representation of a convex set~\cite{gouveia2013lifts}.

\begin{definition}[Conic representation]\label{def:cvxsetKrep}
A convex set $\cvxset \subseteq \sE$ has a conic representation over $K$ if it can be written as $\cvxset = \{\Pi(x) : x \in K,\; F(x) = g\}$
where $\Pi:\operatorname{linspan}(K) \to \sE$ and $F:\operatorname{linspan}(K)\to \sF$ are linear maps and $g \in \sF$. The size of the representation is $\dim(K)$.
\end{definition}

Clearly, if a convex set $\cvxset$ has a conic representation over $K$, then any linear optimization problem over $\sC$ can be expressed as a conic program over $K$. 
We recall below the standard definition of a semidefinite representation of a convex set~\cite[Definition 1]{gouveia2013lifts}. 

\begin{definition}[SDP representation]\label{def: sdp representation definition}
A convex set $\cvxset \subseteq \sE$ has a SDP representation of size\footnote{We adopt this convention for consistency with Definition \ref{def:cvxsetKrep}. This convention is different from existing conventions in the literature where the size of the representation would be $s$ instead of $s^2$.}
$s^2$ if it can be written as $\cvxset = \left\{\Pi(x): x \in \PSD(\CC^s),\, F(x) = g\right\}$,
where $\Pi: \HERM(\CC^s) \to \sE$ and $ F:\HERM(\CC^s) \to \sF$ are linear maps
and $g \in \sF$. 
\end{definition}

\begin{definition}[Quantum relative entropy program]
A quantum relative entropy program is a conic program over a Cartesian product of positive semidefinite cones, quantum relative entropy cones, and operator relative entropy cones. The quantum relative entropy cone is defined as~\cite{chandrasekaran2017relative,he2024exploiting,he2024qics}:
\begin{align}
\cK_{\text{qre}}(\CC^n) := \text{cl} \{(A,B,t) \in \PSD(\CC^n) \times \PSD(\CC^n) \times \RR : A \ll B,\, D(A\|B) \leq t\},
\end{align}
where $A \ll B$ indicates that the support of $A$ is included in the support of $B$. The operator relative entropy cone is defined as:
\begin{align}
\cK_{\text{ore}}(\CC^n) := \text{cl} \{(A,B,T) \in \PD(\CC^n) \times \PD(\CC^n) \times \HERM(\CC^n) : A \ll B,\, D_{\operatorname{op}}(A\|B) \leq T\},
\end{align}
with $D_{\operatorname{op}}(A\|B) := A^{1/2} \log(A^{1/2} B^{-1} A^{1/2}) A^{1/2}$ being the operator relative entropy.
\end{definition}

The following well-known proposition will be useful later, see e.g., 
\cite[Lemma 4.1.8]{netzer2012spectrahedra}. We include a sketch for convenience.

\begin{lemma}\label{lem: SDP representation for dual set}
Let $\cvxset$ be a convex set in some Euclidean space $\sE$. If $\cvxset$ has an SDP representation of size $s^2$, then the epigraph of its support function
\begin{align}
   \epi(h_{\cvxset}) := \{(w, t) \in \sE \times \RR : h_{\cvxset}(w) \leq t\}.
\end{align}
has an SDP representation of size $s^2+1$.
\end{lemma}
\begin{proof}[Proof sketch.]
Assume $\cvxset$ has an SDP representation as in Definition~\ref{def: sdp representation definition} of size $s^2$. Given $w \in \sE$ we have
\begin{align}
h_{\cvxset}(w) &:= \sup_{z \in \cvxset} \< w , z \> \\
		   &=  \sup \{ \< w , \Pi(Y) \> : Y \in \PSD(\CC^s), F(Y) = g\} \label{eq: support function SDP}\\
		   &= \inf_{\lambda \in \sF} \{ \<\lambda, g\> : F^*(\lambda) - \Pi^*(w) \geq 0 \}
\end{align}
where $F^*$ and $\Pi^*$  denote respectively the adjoint maps of $F:\HERM(\CC^s)\to \sF$ and $\Pi:\HERM(\CC^s) \to \sE$, and the last step follows from SDP duality (assuming that the SDP in~\eqref{eq: support function SDP} is strictly feasible).
This shows that 
\[
\epi(h_{\sC}) = \{(w,t) \in \sE \times \RR : \exists \lambda \in \sF \text{ s.t. } \<\lambda, g\> \leq t \text{ and } F^*(\lambda) - \Pi^*(w) \in \PSD(\CC^s)\}
\]
which is a semidefinite representation of $\epi(h_{\sC})$ of size $s^2+1$.
\end{proof}

\subsection{Efficient approximation by symmetry reduction}

With the above definitions, we now show that the converging bounds for $D^\reg(\sA\|\sB)$, derived from \cite{fang2024generalized} under structural assumptions on the sets $\sA_n$ and $\sB_n$, allow for an efficient approximation of $D^\reg(\sA\|\sB)$ using quantum relative entropy programs. We first recall the following lemma from \cite{fang2024generalized}.

\begin{lemma}\label{lem: generalized AEP finite estimate}~\cite[Lemma 28, 29]{fang2024generalized}
Let $\cH$ be a Hilbert space of finite dimension $d$. Let $\{\sA_n\}_{n\in\NN}$ and $\{\sB_n\}_{n\in\NN}$ be two sequences of sets satisfying Assumption~\ref{ass: steins lemma assumptions} and $\sA_n \subseteq \density(\cH^{\ox n})$, $\sB_n \subseteq \PSD(\cH^{\ox n})$ and $D_{\max}(\sA_n\|\sB_n) \leq cn$, for all $n \in \NN$ and a constant $c \in \RR_{\pl}$. Then the regularized relative entropy $D^{\reg}(\sA \| \sB)$ can be estimated using the following bounds: 
\begin{align}\label{eq: generalized AEP finite estimate}
    &\frac{1}{m} D_{\Meas}(\sA_m \| \sB_m)\leq D^{\reg}(\sA \| \sB) \leq \frac{1}{m} D(\sA_m \| \sB_m), \quad \forall m \geq 1,
\end{align}
with explicit convergence guarantees
\begin{align}\label{eq: generalized AEP finite estimate convergence}
        \frac{1}{m} D(\sA_m \| \sB_m) - \frac{1}{m} D_{\Meas}(\sA_m \| \sB_m) \leq 
        \frac{1}{m} 2(d^2 + d) \log (m+d).
\end{align} 
\end{lemma}

\begin{shaded}
\begin{proposition}(Approximation of regularized relative entropies.)\label{prop: relative entropy program}
    Assume the same conditions for $\{\sA_n\}_{n\in \NN}$ and $\{\sB_n\}_{n\in \NN}$ as in Lemma~\ref{lem: generalized AEP finite estimate}. If each $\sA_m$ and $\sB_m$ have SDP representations of size $s^{2m}$, then $D(\sA_m \| \sB_m)$ and $D_{\Meas}(\sA_m \| \sB_m)$ can be computed by quantum relative entropy programs of size $O(k^{2m})$ with $k=\max\{s,d\}$.  
    
    As a result, we can approximate $D^{\reg}(\sA \| \sB)$ within additive error $\delta$ by a quantum relative entropy program of size $O(k^{2m_0})$ with $m_0 = \lceil \frac{8d^2}{\delta}\log \frac{d^2}{\delta}\rceil$ using Eq.~\eqref{eq: generalized AEP finite estimate convergence}. 
\end{proposition}
\end{shaded}
\begin{proof}
Suppose $\sA_m$ and $\sB_m$ have SDP representations as
\begin{align}
    \sA_m & = \left\{\Pi_m(X): X \in \PSD(\CC^{s^m}), F_m(X) = g_m\right\},\\
    \sB_m & = \left\{\Pi'_m(X'): X' \in \PSD(\CC^{s^m}), F'_m(X') = g'_m\right\}.
\end{align}
Then for the quantum relative entropy, we can write
\begin{align}
    D(\sA_m\|\sB_m) & = \inf \left\{t: \rho_m \in \sA_m, \sigma_m \in \sB_m, t \geq D(\rho_m\|\sigma_m)\right\}\\
    & = \inf \left\{t: \rho_m \in \sA_m, \sigma_m \in \sB_m, (\rho_m, \sigma_m, t) \in \cK_{\text{qre}}\right\} \\
    &= \inf \Big\{t: \rho_m = \Pi_m(X), \sigma_m = \Pi'_m(X'), F_m(X) = g_m, F'_m(X') = g'_m \\
    &\quad ((\rho_m, \sigma_m, t), X, X') \in \cK_{\text{qre}} \times \PSD(\CC^{s^m}) \times \PSD(\CC^{s^m}) \Big\}
\end{align}
which is a quantum relative entropy program of size $2d^{2m}+2s^{2m}+1 = O(k^{2m})$. For the measured relative entropy, we have from~\cite[Lemma 19]{fang2024generalized} that
\begin{align}
    D_{\Meas}(\sA_m\|\sB_m) = \sup_{W_m \in (\sB_m)^\circ_{\pl \pl}} -h_{\sA_m}(-\log W_m).
\end{align}
We now proceed to write the convex program more explicitly. Since $\sA_m \subseteq \PSD$, we know that $h_{\sA_m}$ is operator monotone (i.e., $h_{\sA_m}(X) \leq h_{\sA_m}(Y)$ if $X\leq Y$). This gives
\begin{align}\label{eq: DM general convex program proof tmp1}
D_{\Meas}(\sA_m\| \sB_m) = 
\sup_{W_m, V_m} \left\{ -h_{\sA_m}(V_m) :  W_m\in \polarPD{(\sB_m)}, V_m \geq -\log W_m \right\}.
\end{align}
Noting that $h_{\sA_m}(V_m + t'I_m) = h_{\sA_m}(V_m) + t'$ as $\sA_m \subseteq \density$, we have
\begin{align}
    h_{\sA_m}(V_m) & = \inf_{t'} \left\{h_{\sA_m}(V_m + t'I_m) - t' : V_m + t'I_m \geq 0\right\}\\
    & = \inf_{t,t'} \left\{ t - t' : t > 0, V_m + t'I_m \in t{\sA_m^\circ}, V_m+t'I_m \geq 0 \right\}.
\end{align}
Taking this into Eq.~\eqref{eq: DM general convex program proof tmp1}, we have 
\begin{align}\label{eq: DM Am Bm variational}
    D_{\Meas}(\sA_m \| \sB_m) 
    =  \sup_{W_m,V_m,t,t'} \Big\{&t'-t: t > 0, (W_m, 1) \in \epi(h_{\sB_m}),\\ & (V_m + t'I_m, t) \in \epi(h_{\sA_m}), (I_m, W_m, V_m) \in \cK_{\text{ore}}\Big\} \notag
\end{align}
which is also a quantum relative entropy program of size $2s^{2m}+2+3d^{2m} = O(k^{2m})$ by Lemma~\ref{lem: SDP representation for dual set}. 
\end{proof}

The above result shows that the regularized relative entropy $D^\reg(\sA\|\sB)$ can be estimated using quantum relative entropy programs but with size exponential in the dimension $\max\{d,s\}$. In the following, we aim to exploit the permutation invariance of $\sA_m$ and $\sB_m$ to reduce the complexity and make the estimation more efficient. For this, we first show that the optimal solution of the relative entropy programs can be restricted to permutation invariant states.

Let $k \in \NN$ be a fixed integer and $\cH$ be a finite-dimensional vector space with $\dim(\cH) = d$. Then the natural action of the symmetric group $\mathfrak{S}_k$ on $\cH^{\ox k}$ by permuting the indices is given by
\begin{align}
    \pi \cdot (h_1\ox \cdots \ox h_k) = h_{\pi^{-1}(1)} \ox \cdots \ox h_{\pi^{-1}(k)},
\end{align}
for any $h_i \in \cH$ and $\pi \in \mathfrak{S}_k$.  Let $P_\pi$ be the permutation operator corresponding to the action of $\pi$ on the suitable space. Denote the twirling operation as $\cT_k(X): = \frac{1}{\left|\mathfrak{S}_k\right|} \sum_{\pi \in \mathfrak{S}_k} P_{\pi} X P_{\pi}^\dagger$.
The algebra of $\mathfrak{S}_k$-invariant operators on $\cH^{\ox k}$ is denoted by 
\begin{align}
    \sI_k := \left\{X \in \sL(\cH^{\ox k}): P_\pi X P_\pi^\dagger = X, \forall \pi \in \mathfrak{S}_k\right\}.
\end{align}

\begin{shaded}
\begin{lemma}\label{lem: permutation invariance restriction}
Let $\sA_m \subseteq \density(\cH^{\ox m})$ and $\sB_m \subseteq \PSD(\cH^{\ox m})$ be convex, compact, permutation invariant sets. Then we have that
\begin{align}
    D(\sA_m \|\sB_m) & =\ \; D(\sA_m \cap \sI_m \| \sB_m \cap \sI_m),\\
    D_\Meas(\sA_m \|\sB_m) & = D_\Meas(\sA_m \cap \sI_m \| \sB_m \cap \sI_m).
\end{align}
\end{lemma}
\end{shaded}
\begin{proof}
The direction of ``$\leq$'' is clear for both equations. We now show the other direction. For the case of quantum relative entropy, we have that for any $\rho_m \in \sA_m$ and $\sigma_m \in \sB_m$,
\begin{align}
    D(\rho_m\|\sigma_m) \geq D(\cT_m(\rho_m)\|\cT_m(\sigma_m)) \geq D(\sA_m \cap \sI_m \| \sB_m \cap \sI_m),
\end{align}
where the first inequality follows by the data-processing inequality of quantum relative entropy and the second inequality follows because $\cT_m(\rho_m) \in \sA_m \cap \sI_m$ and $\cT_m(\sigma_m) \in \sB_m \cap \sI_m$ by the permutation invariance of $\sA_m$ and $\sB_m$. Optimizing $\rho_m \in \sA_m$ and $\sigma_m \in \sB_m$ on both sides, we get $D(\sA_m \|\sB_m) \geq D(\sA_m \cap \sI_m \| \sB_m \cap \sI_m)$. The case of measured relative entropy follows by the same argument.
\end{proof}

The following result gives the SDP representation of the intersection of $\sA_m$ and $\sI_m$.

\begin{shaded}
    \begin{lemma}\label{lem: permutation invariance restriction 1}
    If the SDP representation of $\sA_m$ in Eq.~\eqref{def: sdp representation definition} satisfies the symmetry conditions: (1) $\Pi_m(P_\pi X P_\pi^\dagger) = P_\pi \Pi_m(X) P_\pi^\dagger$; (2) $F_m(P_\pi X P_\pi^\dagger) = P_\pi F_m(X) P_\pi^\dagger$; and (3) $P_\pi g_m P_\pi^\dagger = g_m$ for any $X \in \PSD((\CC^s)^{\ox m})$ and any permutation operator $P_\pi$ on the suitable spaces, then 
    \begin{align}\label{eq: permutation invariance reduction tmp1}
        \sA_m \cap \sI_m = \left\{\Pi_m(X): X \in \PSD((\CC^s)^{\ox m}) \cap \sI_m, F_m(X) = g_m\right\}.
    \end{align}
    \end{lemma}
    \end{shaded}
    \begin{proof}
    We consider the inclusion ``$\supseteq$'' first. For any $\Pi_m(X)$ with $X \in \PSD((\CC^s)^{\ox m}) \cap \sI_m$ and $F_m(X) = g_m$, it is clear that $\Pi_m(X)\in \sA_m$ and moreover, we have 
    \begin{align}
        P_\pi \Pi_m(X) P_\pi^\dagger = \Pi_m(P_\pi X P_\pi^\dagger) = \Pi_m(X),
    \end{align}
    where the first equality follows from the assumption of $\Pi_m$ and the second equality follows as $X \in \sI_m$. As this holds for any $\pi \in \mathfrak{S}_m$, we have $\Pi_m(X) \in \sI_m$ and therefore $\Pi_m(X) \in \sA_m \cap \sI_m$. Now we prove the inclusion ``$\subseteq$''. For any $\Pi_m(X) \in \sI_m$ with $X \in \PSD((\CC^s)^{\ox m}), F_m(X) = g_m$, we have 
    \begin{align}
        \Pi_m(\cT_m(X)) = \cT_m(\Pi_m(X)) = \Pi_m(X),
    \end{align}
    where the first equality follows from the linearity and permutation invariant assumption of $\Pi_m$ and the second equality follows as $\Pi_m(X) \in \sI_m$. Similarly, we can argue that $F_m(\cT_m(X)) = \cT_m(g_m) = g_m$. This implies that $\Pi_m(X)$ belongs to the set on the right hand side of Eq.~\eqref{eq: permutation invariance reduction tmp1} because $\cT(X) \in \PSD((\CC^s)^{\ox m}) \cap \sI_m$. This concludes the proof.
    \end{proof}

Next, we apply representation theory to reduce the size of the SDP representation for $\sA_m \cap \sI_m$. For that, we follow the notation in~\cite{Fawzi_2022}. The set $\sI_m$ of operators acting on $\cH$ that are invariant under permutation is isomorphic to a direct sum of matrix algebras. In order to describe the optimization problems, we have to describe this isomorphism explicitly. For that, let $\Par(d,m)$ be the set of partitions $\lambda$ of $m$ of height $d$ (i.e., $\lambda_1 \geq \dots \geq \lambda_d > 0$ with $\lambda_1 + \dots + \lambda_d = m$), $T_{\lambda, d}$ be the set of semistandard $\lambda$-tableaux with entries in $[d]$ and $m_{\lambda}^{\cH} = |T_{\lambda, d_{\cH}}|$ (see~\cite{Fawzi_2022} for more details on these concepts). We then define the map
\begin{align}
    \phi_\cH:  \sI_m(\cH^{\ox m}) & \to \bigoplus_{\lambda \in \Par(d_\cH,m)} \CC^{m_\lambda^\cH \times m_\lambda^\cH}\quad \text{with}\quad X  \mapsto \bigoplus_{\lambda \in \operatorname{Par}\left(d_{\mathcal{H}}, m\right)}\left(\left\langle X u_\gamma, u_\tau\right\rangle\right)_{\tau, \gamma \in T_{\lambda, d_{\mathcal{H}}}},\label{eq: bijection block digonalization}
\end{align}
where $\{u_{\tau}\}_{\tau \in T_{\lambda, d_{\cH}}}$ are vectors in $\cH^{\otimes m}$ the exact definition of which can be found in~\cite{Fawzi_2022}; see also~\cite[Section 2.1]{litjens2017semidefinite} for more details.

Note that in this decomposition, the number of blocks and the size of the blocks are bounded by a polynomial in $m$. In particular, we have
\begin{align}
t^{\mathcal{H}} & :=\left|\operatorname{Par}\left(d_{\mathcal{H}}, m\right)\right| \leq(m+1)^{d_{\mathcal{H}}}, \label{eq: num of blocks}\\
m_\lambda^{\mathcal{H}} & :=\left|T_{\lambda, d_{\mathcal{H}}}\right| \leq(m+1)^{d_{\mathcal{H}}\left(d_{\mathcal{H}}-1\right) / 2}, \quad \forall \lambda \in \operatorname{Par}\left(d_{\mathcal{H}}, m\right). \label{eq: block size}
\end{align}
Therefore, we get the dimension of the permutation-invariant subspace as
\begin{align}
m^{\mathcal{H}}:=\operatorname{dim}\left[\sI_m\right] \leq(m+1)^{d_{\mathcal{H}}^2}.    
\end{align}

With the SDP representation of $\sA_m\cap \sI_m$ in Lemma~\ref{lem: permutation invariance restriction 1}, we can now apply the linear map $\phi_\cH$ in Eq.~\eqref{eq: bijection block digonalization} to decompose the operator $X$ on the exponentially large space into a block-diagonal form. Specifically, let $\cH_1 = \CC^s$, $\cH_2 = \CC^d$ and $\cH_3 = \CC^f$, and define
\begin{align}
    \Pi_m': \bigoplus_{\lambda \in \Par(s,m)} \CC^{m_\lambda^{\cH_1} \times m_\lambda^{\cH_1}} & \to \bigoplus_{\lambda \in \Par(d,m)} \CC^{m_\lambda^{\cH_2} \times m_\lambda^{\cH_2}}\quad \text{with} \quad 
    X  \mapsto \ \ \ \phi_{\cH_2} (\Pi_{m}(\phi_{\cH_1}^{-1}(X))),
\end{align}
and
\begin{align}
    F_m': \bigoplus_{\lambda \in \Par(s,m)} \CC^{m_\lambda^{\cH_1} \times m_\lambda^{\cH_1}} & \to \bigoplus_{\lambda \in \Par(f,m)} \CC^{m_\lambda^{\cH_3} \times m_\lambda^{\cH_3}} \quad \text{with} \quad
    X  \mapsto \ \ \ \phi_{\cH_3} (F_m(\phi_{\cH_1}^{-1}(X))),
\end{align}
and the linear operator $g_m' = \phi_{\cH_3}(g_m)$. With these notations, we have the following SDP representation.

\begin{shaded}
\begin{lemma}\label{lem: permutation invariance restriction 2}
Let $\cH_1 = \CC^s$ and $\cH_2 = \CC^d$. Then the SDP representation in Eq.~\eqref{eq: permutation invariance reduction tmp1} gives
\begin{align}\label{eq: set reduction tmp3}
    \phi_{\cH_2}(\sA_m \cap \sI_m) = & \Bigg\{\Pi_m'\Bigg(\bigoplus_{\lambda \in \Par(s,m)} X_{\lambda}\Bigg):\\
   &  \forall \lambda \in \Par(s,m), \; X_{\lambda} \in
\PSD\left(\CC^{m_\lambda^{\cH_1}}\right),\; F_m'\Bigg(\bigoplus_{\lambda \in \Par(s,m)} X_{\lambda}\Bigg) = g_m'\Bigg\},\notag
\end{align}
where the SDP representation on the right hand side is of size at most $(m+1)^{s^2}$.
\end{lemma}
\end{shaded}
\begin{proof}
We show the inclusion ``$\subseteq$'' first. For any element $\phi_{\cH_2}(\Pi_m(X)) \in \phi_{\cH_2}(\sA_m \cap \sI_m)$, we have $X \in \PSD((\CC^s)^{\ox m}) \cap \sI_m$ and $F_m(X) = g_m$. Let $X' = \phi_{\cH_1}(X)$. Then $X' \in \PSD(\cH_1^{\ox m})$ and $X' \in \bigoplus_{\lambda \in \Par(s,m)} \CC^{m_\lambda^{\cH_1} \times m_\lambda^{\cH_1}}$. Moreover, $\phi_{\cH_2}(\Pi_m(X)) = \Pi_m'(X')$ and $F_m'(X') = g_m'$. This implies that $\phi_{\cH_2}(\Pi_m(X))$ is included in the right-hand side of Eq.~\eqref{eq: set reduction tmp3}. 

Now we show the other inclusion ``$\supseteq$''. For any element $\Pi_m'(X)$ with $X \in \PSD(\cH_1^{\ox m})$, $X \in \bigoplus_{\lambda \in \Par(d,m)} \CC^{m_\lambda^{\cH_1} \times m_\lambda^{\cH_1}}$ and  $F_m'(X) = g_m'$, let $X'' = \phi^{-1}_{\cH_1}(X)$. Then we have $X'' \in \PSD((\CC^d)^{\ox m}) \cap \sI_m$, $\Pi_m'(X) = \phi_{\cH_2}(\Pi_m(X''))$ and $F_m(X'') = g_m$. This implies $\Pi_m'(X) \in \phi_{\cH_2}(\sA_m \cap \sI_m)$ and concludes the proof.
\end{proof}

Combining Lemmas~\ref{lem: permutation invariance restriction},~\ref{lem: permutation invariance restriction 1} and~\ref{lem: permutation invariance restriction 2}, we have
\begin{align}
    D(\sA_m\|\sB_m)  &= D(\phi_{\cH_2}(\sA_m\cap \sI_m)\|\phi_{\cH_2}(\sB_m\cap \sI_m)),\\
    D_{\Meas}(\sA_m\|\sB_m) & = D_{\Meas}(\phi_{\cH_2}(\sA_m\cap \sI_m)\|\phi_{\cH_2}(\sB_m\cap \sI_m)).
\end{align}
Together with Proposition~\ref{prop: relative entropy program}, we get the efficient approximation of regularized relative entropies presented in Theorem~\ref{thm: efficient relative entropy program}.

\section{Applications}
\label{sec: applications}

In this section, we apply the idea of efficient approximation to several quantum information processing tasks. Note that the applicability of our approximation relies on the structural assumptions of the sets in Assumption~\ref{ass: steins lemma assumptions}, which holds directly for many cases such as the singleton set, the set of incoherent states used in coherence theory and the image set of a channel used in adversarial channel discrimination. We exemplify the last case in Section~\ref{sec: Adversarial quantum channel discrimination}.  In cases where the task of interest does not directly satisfy Assumption~\ref{ass: steins lemma assumptions}, one can follow a general methodology of relaxing the set of interest to one that does, particularly regarding the polar assumption in (A.4). For instance, in entanglement theory, the set of separable states can be relaxed to the Rains set, which satisfies all necessary assumptions. Similarly, in fault-tolerant quantum computing, the set of stabilizer states can be relaxed to the set of states with non-positive mana, which also fulfills the required conditions.  We provide several examples in Sections~\ref{sec: Entanglement cost for quantum states and channels},~\ref{sec: Quantum entanglement distillation} and~\ref{sec: Magic state distillation} to illustrate this idea, which gives improvement to the state-of-the-art results in the literature.

\subsection{Adversarial quantum channel discrimination}
\label{sec: Adversarial quantum channel discrimination}

We now apply the general theory in Theorem~\ref{thm: efficient relative entropy program} to compute the minimum output channel divergence, which serves as the key quantity in adversarial quantum channel discrimination~\cite{fang2025adversarial}.

\begin{definition}(Minimum output quantum channel divergence.)
    Let $\DD$ be a quantum divergence between quantum states. Let $\cN\in \CPTP(A:B)$ and $\cM \in \CP(A:B)$. Define the corresponding minimum output channel divergence by
    \begin{align}\label{eq: minimum output channel divergence}
    \DD^{\inf}(\cN\|\cM)&:= \inf_{\substack{\rho\in \density(A)\\ \sigma\in \density(A)}} \DD(\cN_{A\to B}(\rho_{A})\| \cM_{A \to B}(\sigma_{A})). 
    \end{align}
    Define its regularized channel divergence by
    \begin{align}
        \DD^{\inf,\reg}(\cN\|\cM):= \lim_{n\to \infty} \frac{1}{n} \DD^{\inf}(\cN^{\ox n}\|\cM^{\ox n}).
    \end{align}
\end{definition}

The following result shows that the regularized channel divergence can be efficiently approximated by quantum relative entropy programs.

\begin{shaded}
\begin{corollary}(Efficient approximation of the regularized minimum output channel divergences.)
Let $\cN \in \CPTP(A:B)$ and $\cM \in \CP(A:B)$ with $\dim A = s$ and $\dim B = d$. The minimum output quantum channel divergences $D({\cN^{\ox m}\|\cM^{\ox m}})$ and $D_{\Meas}({\cN^{\ox m}\|\cM^{\ox m}})$ can both be computed by quantum relative entropy programs of size $O((m+1)^{k^2})$ with $k=\max\{s,d\}$.

As a result, $D^{\inf,\infty}(\cN\|\cM)$ can be approximated within additive error $\delta$ by a quantum relative entropy program of size $O((m_0+1)^{k^2})$ with $m_0 = \lceil \frac{8d^2}{\delta}\log \frac{d^2}{\delta}\rceil$. 
\end{corollary}
\end{shaded}
\begin{proof}
It is clear that the minimum output channel divergence is the divergence between the image sets of the channels, \begin{align}
    \sA_m & = \{\cN^{\ox m}(\rho): \rho \in \PSD{\left((\CC^{s})^{\ox m}\right)}, \tr \rho = 1\},\\
    \sB_m & = \{\cM^{\ox m}(\rho): \rho \in \PSD{\left((\CC^{s})^{\ox m}\right)}, \tr \rho = 1\}.
\end{align} 
These sets satisfy all the required assumptions in Assumption~\ref{ass: steins lemma assumptions}~\cite{fang2025adversarial}. Moreover, $\sA_m = \{\cN^{\ox m}(\rho): \rho \in \PSD{((\CC^{s})^{\ox m})}, \tr \rho = 1\}$ is a SDP representation of size $s^{2m}$ and $\Pi_m = \cN^{\ox m}$, $F_m = \tr$, $g_m = 1$ satisfy the symmetry conditions in Lemma~\ref{lem: permutation invariance restriction 1}. The same holds for $\sB_m$. Applying Theorem~\ref{thm: efficient relative entropy program}, we have the asserted statement.
\end{proof}

In the following, we provide an explicit example to show that the regularized minimum output channel divergence $D^{\infty,\inf}(\cN\|\cM)$ can be approximated by $D^{\inf}(\cN^{\ox m}\|\cM^{\ox m})/m$ from above and $D_{\Meas}^{\inf}(\cN^{\ox m}\|\cM^{\ox m})/m$ from below, with the approximation improving as $m$ increases.

\begin{figure}[!htb]
    \centering
    \includegraphics[width=15cm]{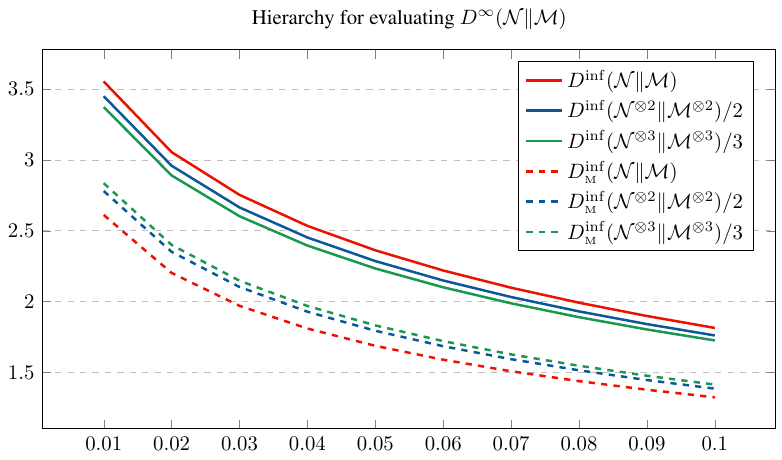}  
    \caption{Estimating the regularized minimum output channel divergences.}
    \label{fig: umegaki_strict_subadditivity}
\end{figure}

This example is given by two qutrit quantum channels. Let $\cN(\cdot) = \tr[\cdot] \ket{\rho}\bra{\rho}$ to be the replacer channel with $\ket{\rho} = (2\ket{0} + \ket{1} + 2\ket{2})/3$. Let $\cM$ be the platypus channel~\cite[Eq.~(170)]{Leditzky_2023}, $\cM(X) = M_0 X M_0^\dagger + M_1 X M_1^\dagger$ with Kraus operators
\begin{align}
    M_0 = \begin{bmatrix}
        \sqrt{p} & 0 & 0\\
        0 & 0 & 0\\
        0 & 1 & 0
    \end{bmatrix},
    \qquad 
    M_1 = \begin{bmatrix}
        0 & 0 & 0\\
        \sqrt{1-p} & 0 & 0\\
        0 & 0 & 1
    \end{bmatrix}.
\end{align}
Since $\cM(I/3) = (p\ket{0}\bra{0} + (1-p)\ket{1}\bra{1} + 2\ket{2}\bra{2})/3$ is a full rank state for all $p \in (0,1)$, the following quantities are finite and can be evaluated by the QICS package~\cite{he2024qics},
\begin{align}
D^{\inf}(\cN^{\ox m}\|\cM^{\ox m}) & = \inf_{\sigma_m \in \density} D(\rho^{\ox m}\|\cM^{\ox m}(\sigma_m)),\\
D_{\Meas}^{\inf}(\cN^{\ox m}\|\cM^{\ox m}) & = \inf_{\sigma_m \in \density} D_{\Meas}(\rho^{\ox m}\|\cM^{\ox m}(\sigma_m)).
\end{align}

The numerical result is given in Figure~\ref{fig: umegaki_strict_subadditivity}. It shows a clear separation between upper bounds with $m=1,2,3$ and $p \in [0.01,0.1]$ and also lower bounds with the same parameter range, confirming the strict subadditivity of the minimum output Umegaki channel divergence and the strict superadditivity of the minimum output measured channel divergence. Moreover, as we increase the number of $m$, the lower and upper bounds provide better approximation to $D^{\infty,\inf}(\cN\|\cM)$.  

\subsection{Entanglement cost for quantum states and channels}
\label{sec: Entanglement cost for quantum states and channels}

The \emph{entanglement cost} of a quantum state, denoted as $E_{C,\Omega}$, is the minimum number of Bell states required to prepare one copy of this state under a class of operations $\Omega$. Of particular interest is the local operation and classical communication (LOCC) operations. It is known that computing $E_{C,\text{LOCC}}$ is NP-hard in general~\cite[Theorem 1]{Huang2014}. Therefore, finding efficiently computable lower and upper bounds to estimate $E_{C,\text{LOCC}}$ is of fundamental importance. Here, we focus on deriving lower bounds, which indicate that no matter what preparation strategies are used, the amount of entanglement consumed cannot be smaller than this value.

There are several lower bounds in the literature, but they are unsatisfactory for different reasons. One such lower bound is given by the regularized PPT-relative entropy of entanglement~\cite[Eq. (8.235)]{hayashi2017quantum},
\begin{align}\label{eq: ec eppt lower bound}
E_{C,\text{LOCC}}(\rho) \geq D^\reg(\rho\|\PPT):= \lim_{n\to \infty} \frac1n D(\rho^{\ox n}\|{\PPT}(A^n:B^n)),
\end{align}
which is difficult to evaluate due to its regularization. Note that the $\PPT$  set does not satisfy the assumption (A.4) and therefore cannot directly apply Theorem~\ref{thm: efficient relative entropy program}. A counter-example can be given by the projector on the $\CC^3\otimes \CC^3$ antisymmetric subspace, denoted as $\rho^a$. Since the support function $h_{\cvxset_{\PPT}}(\cdot)$ is given by a semidefinite program, we find numerically that $h_{{\PPT}}((\rho^a)^{\ox 2}) - h_{{\PPT}}(\rho^a)^2 \approx 0.0093 > 0$.

A single-letter lower bound for entanglement cost is provided by the quantum squashed entanglement \cite{christandl2004squashed}, but its computability remains uncertain due to the unbounded dimension of the extension system. Wang and Duan have proposed two single-letter SDP lower bounds in separate works:
\begin{align}
E_{\WD,1}(\rho_{AB}) & := -\log \max \tr \Pi_{\rho} V_{AB} \quad \ \text{s.t.}\ \|V_{AB}^{\sfT_B}\|_1 = 1,\ V_{AB} \geq 0, \quad \ ~\cite{Wang2017}\\
E_{\WD,2}(\rho_{AB}) & :=  -\log \min \ \ \|Y_{AB}^{\sfT_B}\|_{\infty}
 \quad \ \ \text{s.t.}  \  -Y_{AB} \leq \Pi_\rho^{\sfT_B} \leq Y_{AB}. \quad \ \ ~\cite{Wang2017irreversibility}\label{eq: definition of E eta}
\end{align}
It has been shown that for any $\rho \in \density(AB)$,
\begin{align}\label{eq: Ec lower bound E eta}
 \max\big\{E_{\WD,1}(\rho), E_{\WD,2}(\rho)\big\} \leq D^\reg(\rho\|\PPT) \leq E_{C,\text{LOCC}}(\rho).
\end{align}
By the definition of min-relative entropy and the Rains set, we can write~\footnote{The original definition of $E_{\WD,1}$ imposes the condition $\|V_{AB}^{\sfT_B}\|_1 = 1$. However, it is equivalent to optimize over the condition $\|V_{AB}^{\sfT_B}\|_1 \leq 1$, as the optimal solution can always be chosen at the boundary.}
\begin{align}
E_{\WD,1}(\rho_{AB}) = D_{\min}(\rho_{AB}\|\Rains(A:B)).
\end{align}
After a detailed examination of the dual program for $E_{\WD,2}$, we can also reformulate it in a comparable structure by introducing an appropriate set of operators. 

\begin{lemma}\label{lem: E eta distance formula}
For any $\rho \in \density(AB)$, it holds that
\begin{align}
E_{\WD,2}(\rho_{AB}) = D_{\min}(\rho_{AB}\|\WD(A:B)),
\end{align}
with the set of Hermitian operators defined by \begin{align}
\WD(A:B) := \big\{\sigma \in \HERM(AB) \,:\, \exists\,Y \in \HERM(AB),\ \text{s.t.}\ -Y \leq \sigma^{\sfT_B} \leq Y, \|Y^{\sfT_B}\|_{1} \leq 1\big\}.
\end{align}
\end{lemma}
\begin{proof}
Using the Lagrangian method, we have the dual SDP of $E_{\WD,2}$ as
\begin{gather}
E_{\WD,2}(\rho) = -\log \max \ \tr \Pi_{\rho} (V-F)^{\sfT_B} \quad \text{s.t.} \notag\\
V+F = (W-X)^{\sfT_B},\, \tr (W+X) \leq 1,\, V,F,W,X \geq 0.
\end{gather}
Let $(V-F)^{\sfT_B} = \sigma$, $V+F = Y$.
By the definition of $D_{\min}$ and $\cvxset_\WD$, we have the desired result.
\end{proof}

As discussed above, both SDP bounds $E_{\WD,1}$ and $E_{\WD,2}$ are essentially entanglement measures induced by the min-relative entropy. However, a significant limitation of these bounds is that they vanish for any full-rank state.

\begin{shaded}
\begin{proposition}\label{prop: WD bounds zero}
For any full rank state $\rho \in \density(AB)$, it holds that
\begin{align}
E_{\WD,1}(\rho) = E_{\WD,2}(\rho) = 0.    
\end{align}
\end{proposition}
\end{shaded}
\begin{proof}
If $\rho_{AB}$ is full rank, then $\Pi_\rho = I_{AB}$ and thus 
\begin{align}
E_{\WD,1}(\rho_{AB})  = -\log \max \big\{\tr \sigma_{AB} :\; \sigma_{AB} \geq 0, \|\sigma_{AB}^{\sfT_B}\|_1 \leq 1\big\}.
\end{align}
It is clear that for any feasible solution $\sigma_{AB}$ it holds that $\tr \sigma_{AB} = \tr \sigma_{AB}^{\sfT_B} \leq \|\sigma_{AB}^{\sfT_B}\|_1 \leq 1$. On the other hand, there is a feasible solution $\sigma_{AB} = I_{AB}/|AB|$ such that $\tr \sigma_{AB} = 1$. Thus the maximization is taken at $\tr \sigma_{AB} = 1$ and thus 
$E_{\WD,1}(\rho_{AB}) = 0$. Similarly, given full rank $\rho_{AB}$, it holds that \begin{align}
E_{\WD,2}(\rho_{AB})  = -\log \max \big\{ \tr \sigma_{AB} :\; \sigma, Y \in \HERM(AB), -Y_{AB} \leq \sigma_{AB}^{\sfT_B} \leq Y_{AB}, \|Y^{\sfT_B}\|_1 \leq 1 \big\}.   
\end{align} 
Then for any feasible solutions $\sigma_{AB},Y_{AB}$, it holds that $\tr \sigma_{AB} = \tr \sigma_{AB}^{\sfT_B} \leq \tr Y_{AB} = \tr Y_{AB}^{\sfT_B} \leq \|Y^{\sfT_B}\|_1 \leq 1$. On the other hand, considering the feasible solution $\sigma_{AB} = Y_{AB} = I_{AB}/|AB|$, we have $\tr \sigma_{AB} = 1$. Thus the maximization is taken at $\tr \sigma_{AB} = 1$ and $E_{\WD,2}(\rho_{AB}) = 0$.
\end{proof}

Recently, Wang et al. introduced the $\PPT_k$ set~\cite{wang2023computable} as 
\begin{align}\label{eq: PPTk set}
\PPT_k(A:B):=\left\{\omega_1 \geq 0: \exists \{\omega_i\}_{i=2}^k,\ \text{s.t.}\ \|\omega_i^{\sfT_B}\|_*\leq \omega_{i+1}, \forall i \in [1:k-1], \|\omega_k^{\sfT_B}\|_1 \leq 1\right\},
\end{align}
where $k \in \NN_+$ and $|X|_*\leq Y$ denotes $-Y\leq X \leq Y$. This set turns out to be related to the quantity $\chi_p$ developed in~\cite{Lami2025}. Building on this set, the authors introduced an efficiently computable lower bound for entanglement cost~\cite{wang2023computable},
\begin{align}
  E_{C,\text{LOCC}}(\rho_{AB}) \geq  E_{\text{WJZ}}(\rho_{AB}):= D_{\Sand,1/2}(\rho_{AB}\|{\PPT_k}(A:B)).
\end{align}

In the following, we show that $\PPT_k$ satisfies Assumption~\ref{ass: steins lemma assumptions} and thus we can apply our Theorem~\ref{thm: efficient relative entropy program} to get an improved bound.

\begin{lemma}\label{lem: PPTk assmuption check}
Let $k \geq 2$. The $\PPT_k$ set defined in Eq.~\eqref{eq: PPTk set} satisfies Assumption~\ref{ass: steins lemma assumptions}.
\end{lemma}
\begin{proof}
It is clear to check that $\PPT_k$ satisfies (A.1) and (A.2). Moreover, suppose $-X_1\leq Y_1 \leq X_1$ and $-X_2\leq Y_2 \leq X_2$, we have $-I \leq X_{1}^{-1/2} Y_1 X_{1}^{-1/2} \leq I$ and $-I \leq X_{2}^{-1/2} Y_2 X_{2}^{-1/2} \leq I$
where the inverses are taken on the supports. This gives $-I \leq X_{1}^{-1/2} Y_1 X_{1}^{-1/2} \ox X_{2}^{-1/2} Y_2 X_{2}^{-1/2} \leq I$
which is equivalent to $-X_{1} \ox X_2 \leq Y_1\ox Y_2 \leq X_1\ox X_2$. Using this result, we can check that $\PPT_k$ satisfies (A.3). As for the assumption (A.4), recall the variational characterization of the trace norm $\|X\|_1 = \min\{\tr Y: -Y \leq X \leq Y\}$~\cite[Lemma S1]{lami2024computable}. We have that 
\begin{align}
    \PPT_k(A:B)=\left\{\omega_1 \geq 0: \exists \{\omega_i\}_{i=2}^k,\ \text{s.t.}\ \|\omega_i^{\sfT_B}\|_*\leq \omega_{i+1}, \forall i \in [1:k], \tr \omega_{k+1}\leq  1\right\}.
\end{align}
The support function is $h_{\PPT_k}(\omega) = \sup_{\omega_1 \in \PPT_k} \tr[\omega \omega_1]$. The Lagrange multiplier is given by
\begin{align}
    & \tr[\omega \omega_1] + \sum_{i=1}^k \tr[\alpha_i(\omega_i^{\sfT_B} + \omega_{i+1})] + \sum_{i=1}^k \tr[\beta_i(\omega_{i+1}-\omega_i^{\sfT_B})] + t (1-\tr[\omega_{k+1}])\\
    & = \tr[\omega_1(\omega + \alpha_1^{\sfT_B} - \beta_1^{\sfT_B})] + \sum_{i=2}^k \tr[\omega_i(\alpha_{i-1}+\beta_{i-1}+\alpha_{i}^{\sfT_B}-\beta_i^{\sfT_B})]\\
    & \hspace{8cm} + \tr[\omega_{k+1}(\alpha_k+\beta_k-tI)] + t.\notag
\end{align}
By the SDP duality, we have the support function \begin{align}
    h_{\PPT_k}(\omega) =  \inf \Big\{ t \geq 0: \  \omega \leq \beta_1^{\sfT_B} - \alpha_1^{\sfT_B}, &\alpha_{i-1}+\beta_{i-1} \leq \beta_{i}^{\sfT_B} - \alpha_i^{\sfT_B}, \forall i \in [2:k], \notag\\
    &\alpha_k+\beta_k \leq t I, \alpha_i \geq 0, \beta_i \geq 0, \forall i \in [1:k]\Big\}.
\end{align}
Then for any $\omega_1, \omega_2$, assume their optimal solutions in the dual program are respectively given by $\{t^1,\alpha_i^1,\beta_i^1\}$ and $\{t^2,\alpha_i^2,\beta_i^2\}$. Then we construct
\begin{align}
    t^3 & := t^1t^2,\\
    \alpha_i^3 & := \alpha_i^1\ox \beta_i^2 + \beta_i^1\ox \alpha_i^2,\\
    \beta_i^3 & := \alpha_i^1\ox \alpha_i^2 + \beta_i^1\ox \beta_i^2.
\end{align}
Note that
\begin{align}
    (\alpha_i^1+\beta_i^1)\ox (\alpha_i^2 + \beta_i^2) & = \alpha_i^3 + \beta_i^3,\\
    (\beta_i^1-\alpha_i^1)^{\sfT_B}\ox (\beta_i^2 - \alpha_i^2)^{\sfT_B} & = (\beta_i^3 - \alpha_i^3)^{\sfT_B}.
\end{align}
Then it is clear that $\{t^3, \alpha_i^3,\beta_i^3\}$ is a feasible solution for $h_{\PPT_k}(\omega_1\ox \omega_2)$. This implies that $h_{\PPT_k}(\omega_1\ox \omega_2) \leq t^3 = t^1t^2 = h_{\PPT_k}(\omega_1)h_{\PPT_k}(\omega_2)$, and proves that $\PPT_k$ satisfies assumption (A.4) by Lemma~\ref{lema: polar set and support function}. This completes the proof.
\end{proof}

\begin{shaded}
\begin{theorem}\label{prop: E eta plus set }
Let $\rho \in \density(AB)$ and $k \geq 2$. Let $(*) = \max\left\{E_{\WD,1}(\rho), E_{\WD,2}(\rho), E_{\text{\rm WJZ}}(\rho)\right\}$ represent the previously known bounds. Then it holds that
\begin{align}
(*) \leq D_{\Meas}(\rho\|\PPT_k) \leq D^\reg(\rho\|\PPT_k) \leq D^\reg(\rho\|\PPT) \leq E_{C,\text{\rm LOCC}}(\rho).
\end{align}
Moreover, $D_{\Meas}(\rho\|\PPT_k)$ can be seen as the first level of approximation to $D^\reg(\rho\|\PPT_k)$ and both quantities can be efficiently estimated.
\end{theorem}
\end{shaded}

\begin{proof}
It has been shown that~\cite[Proposition S4]{wang2023computable},
\begin{align}\label{eq: PPTk relation}
\PPT(A:B) \subseteq \PPT_k(A:B) \subseteq \cdots \subseteq \Rains(A:B).
\end{align}
It is also clear from their definitions that $\PPT_2(A:B) =  {\WD}(A:B)$. This implies
\begin{align}\label{eq: set relation}
    \PPT(A:B) \subseteq \PPT_k(A:B) \subseteq \cdots {\PPT_2}(A:B) \subseteq  {\WD}(A:B)\cap {\Rains}(A:B).
\end{align}
Therefore, the first two inequalities of the asserted result follow from the relation of divergences in Lemma~\ref{thm: comparison of quantum divergence} and the relation of sets in Eq.~\eqref{eq: set relation}. The second inequality follows from the superadditivity in~\cite[Lemma 21]{fang2024generalized} and the asymptotic equivalence in~\cite[Lemma 28]{fang2024generalized}. The third inequality follows from the relation in Eq.~\eqref{eq: set relation}. The last equality is known from Eq.~\eqref{eq: Ec lower bound E eta}. The computability of $D_{\Meas}(\rho\|\PPT_k)$ and $D^\reg(\rho\|\PPT_k)$ follows from Lemma~\ref{lem: PPTk assmuption check} and Theorem~\ref{thm: efficient relative entropy program}.
\end{proof}

Following similar arguments in~\cite{wang2023computable}, we can also show that the new measures $D_{\Meas}(\rho\|\PPT_k)$ and $D^\reg(\rho\|\PPT_k)$ satisfy the desired properties such as normalization, faithfulness and (super-) additivity.

Besides the bounds mentioned, which are established on the PPT set, there is another efficiently computable lower bound on $E_{C,\text{LOCC}}$ given by Lami and Regula in~\cite{lami2023no},
\begin{align}
    E_{\text{LR}}(\rho):= \log\sup\left\{\tr X\rho: \|X^{\sfT_B}\|_\infty \leq 1, \|X\|_\infty = \tr[X \rho]\right\}.
\end{align}
However, this bound also vanishes for full rank states~\cite[Eq. (43)]{regula2022functional}.

In the following, we compare our new bounds with previously established ones through several examples, including Isotropic states and Werner states. To the best of our knowledge, the entanglement costs for these states under LOCC operations remain unresolved. Additionally, we use randomly generated quantum states to showcase the broad applicability and improvement of our bound across unstructured quantum states. In all cases, our experiments clearly demonstrate the superiority of our bound (even for the first level of approximation) over existing ones.

\begin{example}(Isotropic states and Werner states.)
The Isotropic state is defined by a convex mixture of the maximally entangled state and its orthogonal complement,
\begin{align}
    \rho_{I, p} := p \ket{\Phi}\bra{\Phi} + \frac{1-p}{d^2-1} (I - \ket{\Phi}\bra{\Phi}),
\end{align}
where $\ket{\Phi} = \frac{1}{\sqrt{d}}\sum_{i=1}^e \ket{ii}$ is the $d$-dimensional maximally entangled state.
The PPT-relative entropy of entanglement for $\rho_{I,p}$ and its regularization are given by~\cite[Theorem 7]{rains1998improved}
\begin{align}
  D^\reg(\rho_{I,p}\|\PPT) = D(\rho_{I,p}\|\PPT) = \begin{cases}
    0 & \text{if} \quad 0 \leq p \leq \frac{1}{d},\\
    \log d + p \log p + (1-p) \log \frac{1-p}{d-1} & \text{if}\quad \frac{1}{d} \leq p \leq 1.
    \end{cases}
\end{align}
The Werner state is defined by a convex mixture of the normalized projectors on the symmetric ($\rho^s$) and anti-symmetric ($\rho^a$) subspaces,
\begin{align}
    \rho_{W,p}:= (1-p)\rho^s + p \rho^a = \frac{1-p}{d(d+1)} (I + S) + \frac{p}{d(d-1)} (I-S),
\end{align}
where $S = \sum_{i,j=1}^e \ket{ij}\bra{ji}$ is the SWAP operator of dimension $d$. The regularized PPT-relative entropy of entanglement for $\rho_{W,p}$ is given by~\cite{audenaert2001asymptotic}
\begin{align}
D^\reg(\rho_{W,p}\|\PPT) = 
\begin{cases}
0, & \text{if}\quad  0 \leq p \leq \frac{1}{2},\\
1- h(p), & \text{if}\quad \frac{1}{2} < p \leq \frac{d+2}{2d},\\
\log \frac{d+2}{d} + (1-p) \log \frac{d-2}{d+2}, & \text{if}\quad p > \frac{d+2}{2d}.
\end{cases}
\end{align}
Note that both the Isotropic states and the Werner states are full-rank states for any $p \in (0,1)$. The bounds $E_{\WD,1}$, $E_{\WD,2}$, and $E_{\text{LR}}$ all reduce to zero. We then compare our new bound $D_{\Meas}(\rho\|\PPT_2)$ with the previously established bounds $E_{\text{WJZ}}$~\cite{wang2023computable} and the analytical bound $D^\reg(\rho\|\PPT)$ in Figure~\ref{fig: comparison iso werner}. It turns out that the first level of approximation $D_{\Meas}(\rho\|\PPT_2)$ already coincides with the analytical bound $D^\reg(\rho\|\PPT)$ for both cases, improving the numerical bounds $E_{\text{WJZ}}$.

\begin{figure}[H]
    \centering
    \includegraphics[width=\linewidth]{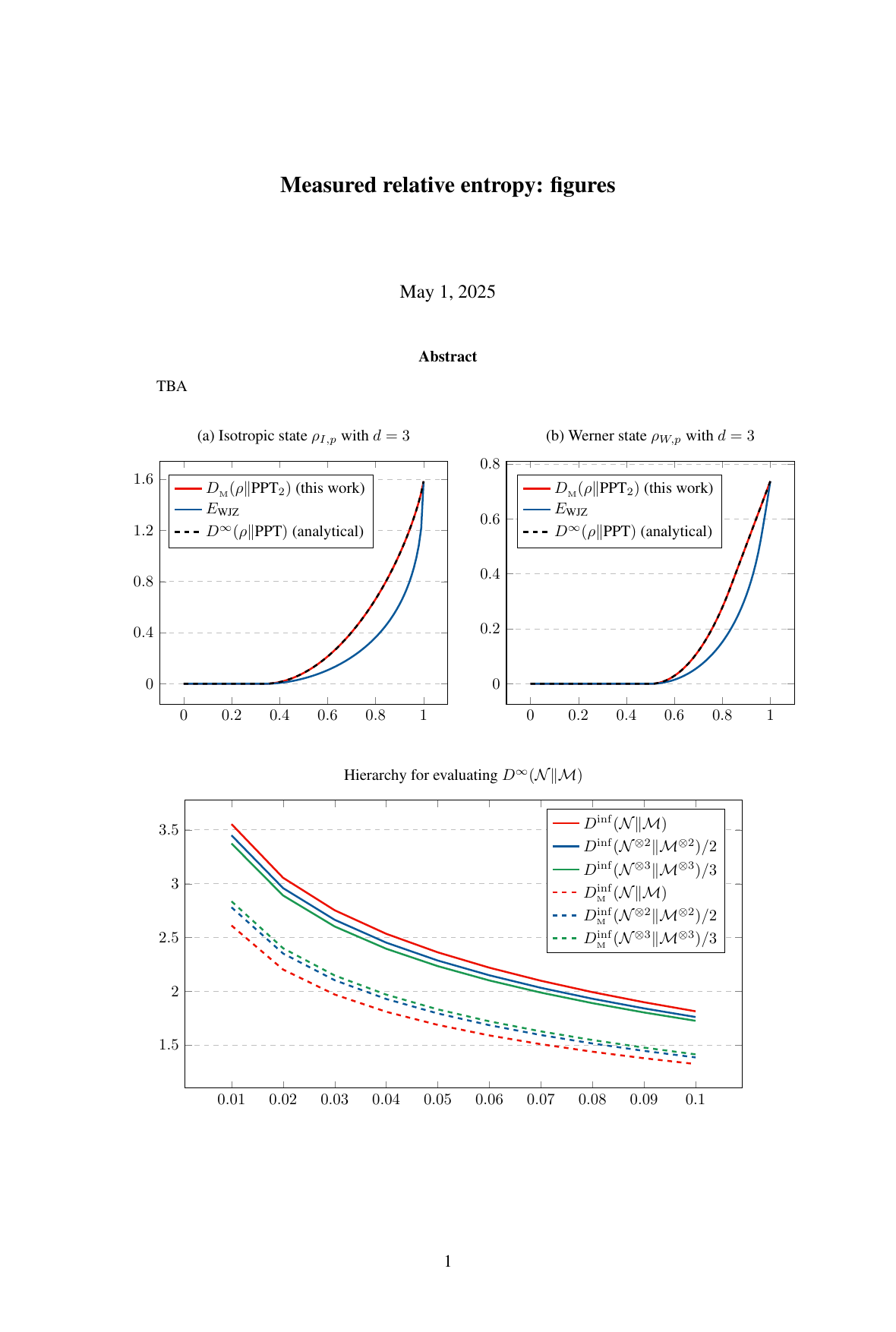}
    \caption{Comparison of the new lower bound $D_{\Meas}(\rho\|\PPT_2)$ with previous bounds $E_{\text{WJZ}}$~\cite{wang2023computable} and $D^\reg(\rho\|\PPT)$~\cite{rains1998improved,audenaert2001asymptotic} for (a) Isotropic states and (b) Werner states. The horizontal axis is the state parameter $p$ and the vertical axis is the value of the entanglement measure.}
    \label{fig: comparison iso werner}
\end{figure}
\end{example}

\begin{example}(Random quantum states.)
Since $D_{\Meas}(\rho\|\PPT_2)$ has been proved to be better than $E_{\text{WJZ}}$ in general, we focus our comparison here with $E_{\text{LR}}$~\cite{lami2023no} by generating random bipartite states according to the Hilbert-Schmidt measure, with varying ranks. For each rank, we generate 500 quantum states of dimension $3 \otimes 3$. The comparison is presented in Figure~\ref{fig:random_compare}. It is evident that $D_{\Meas}(\rho\|\PPT_2)$ outperforms $E_{\text{LR}}$ in most cases, particularly for higher-rank states.

\begin{figure}[H]
    \centering
    \includegraphics[width=12cm]{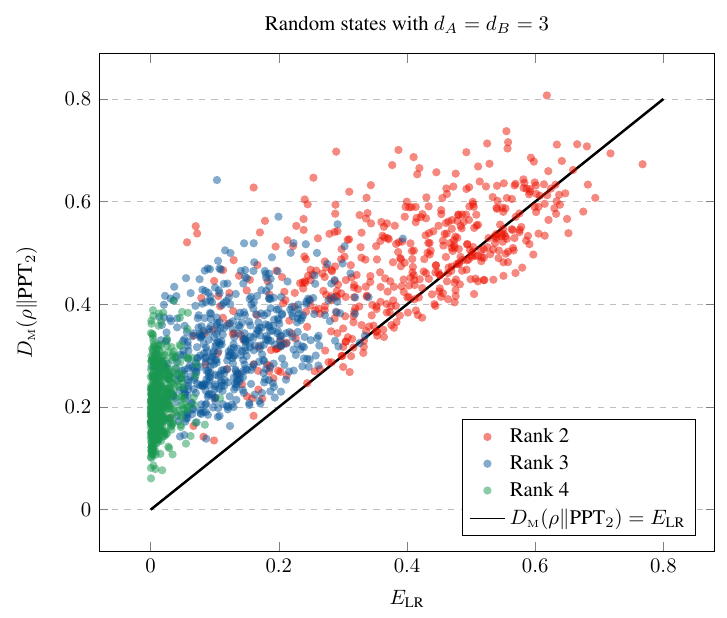}
    \caption{Comparison of the new bound $D_{\Meas}(\rho\|\PPT_2)$ with the previously known bound $E_{\text{LR}}$~\cite{lami2023no} for randomly generated quantum states with different ranks.}
    \label{fig:random_compare}
\end{figure}    
\end{example}

Similar to the entanglement cost for quantum states, the entanglement cost of a quantum channel, denoted by $E_{C,\text{LOCC}}(\cN)$, represents the minimal rate at which entanglement (between the sender and receiver) is required to simulate multiple copies of the channel, given the availability of free classical communication. It is known that~\cite{berta2013entanglement}
\begin{align}
    E_{C,\text{LOCC}}(\cN) \geq \sup_{\rho \in \density(AA')} E_{C,\text{LOCC}}(\cN_{A\to B}(\rho_{AA'})),
\end{align}
where system $\cH_{A'}$ is isomorphic to system $\cH_A$. Wang et al. introduced a lower bound for the entanglement cost of a quantum channel in~\cite{wang2023computable},
\begin{align}
    E_{C,\text{LOCC}}(\cN) \geq E_{\text{WJZ}}(\cN) \geq E_{\text{WJZ}}(\cN_{A\to B}(\Phi_{AA'}))
\end{align}
where $E_{\text{WJZ}}(\cN):= \sup_{\rho \in \density(AA')} E_{\text{WJZ}}(\cN_{A\to B}(\rho_{AA'}))$ and $\Phi_{AA'}$ is the maximally entangled state. This lower bound has been used to demonstrate that the resource theory of entanglement is irreversible for amplitude damping channels. 

Specifically, the amplitude damping channel is defined by
\begin{align}
    \cN_{\rm ad}(\rho) = E_0 \rho E_0^\dagger + E_1 \rho E_1^\dagger,
\end{align}
with Kraus operators $E_0 = \ket{0}\bra{0} + \sqrt{1-\gamma}\ket{1}\bra{1}$ and $E_1 = \sqrt{\gamma}\ket{0}\bra{1}$. Its quantum capacity, i.e., the maximal rate at which entanglement can be generated from the channel, is known as~\cite{giovannetti2005information} 
\begin{align}
    Q(\cN_{\rm ad}) = \max_{p \in [0,1]} h_2((1-\gamma)p) - h_2(\gamma p),
\end{align}
where $h_2$ is the binary entropy. It has been shown in~\cite{wang2023computable} that for $0.25 \lesssim \gamma < 1$,
\begin{align}
    E_{C,\text{LOCC}}(\cN_{\rm ad}) \geq E_{\text{WJZ}}(\cN_{\rm ad}(\Phi_{AA'})) > Q(\cN_{\rm ad}),
\end{align}
Here, we can improve this bound by 
\begin{align}
    E_{C,\text{LOCC}}(\cN_{\rm ad}) \geq D_{\Meas}(\cN_{\rm ad}(\Phi_{AA'})\|\PPT_2) > Q(\cN_{\rm ad})
\end{align}
and show that the gap exists across the entire parameter region $0 < \gamma < 1$ in Figure~\ref{fig:adchannel_irreversibility}.

\begin{figure}[H]
    \centering
    \includegraphics[width=12cm]{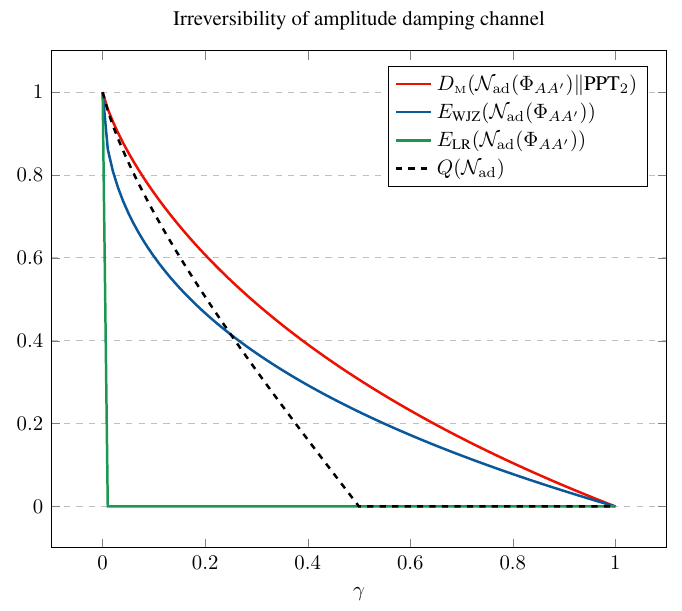}
    \caption{Comparison of the new lower bound $D_{\Meas}(\cN_{\rm ad}(\Phi_{AA'})\|\PPT_2)$ with the previously known bounds $E_{\text{WJZ}}$~\cite{wang2023computable}, $E_{\text{LR}}$~\cite{lami2023no} and the quantum capacity $Q$ for amplitude damping channel. The horizontal axis is the channel parameter $\gamma$ and the vertical axis is the value of the entanglement measure.}
    \label{fig:adchannel_irreversibility}
\end{figure}

\subsection{Quantum entanglement distillation}
\label{sec: Quantum entanglement distillation}

Entanglement distillation is an essential quantum information processing task in quantum networks that involves converting multiple copies of noisy entangled states into a smaller number of Bell states. The \emph{distillable entanglement} of a bipartite state $\rho_{AB}$, denoted by $E_{D,\Omega}(\rho_{AB})$, represents the maximum number of Bell states that can be extracted from the given state with asymptotically vanishing error under the operation class $\Omega$. It has been established that the distillable entanglement under asymptotically non-entanglement generating operations, denoted by $E_{D,\text{\rm ANE}}$, is given by the \emph{regularized relative entropy of entanglement}~\cite{Brand_o_2010},
\begin{align}
    E_{D,\text{\rm ANE}}(\rho_{AB}) = D^\reg(\rho_{AB}\|\SEP):= \lim_{n\to \infty} \frac{1}{n} D(\rho_{AB}^{\ox n}\|\SEP(A^n:B^n))
\end{align}
where $\SEP(A:B)$ denotes the set of all separable states between $\cH_A$ and $\cH_B$. Evaluating this quantity is hard in general, as it involves a limit as well as the separability problem, which is known to be computationally hard~\cite{gurvits2003classical}.

As $D^\reg(\rho_{AB}\|\SEP)$ is a minimization problem, any feasible solution gives an upper bound. Here, we can derive an efficient lower bound for \( D^\reg(\rho_{AB} \| \SEP) \) by relaxing $\SEP$ to
the Rains set~\cite{rains2001semidefinite,audenaert2002asymptotic}~\footnote{The set of PPT states does not satisfy Assumption~\ref{ass: steins lemma assumptions}, see discussion after Eq.~\eqref{eq: ec eppt lower bound}.} 
\begin{align}
\Rains(A:B):= \left\{\sigma \in \PSD(AB): \|\sigma^{\sfT_B}\|_1 \leq 1\right\}.
\end{align}
We can show that it satisfies Assumption~\ref{ass: steins lemma assumptions}. This is because we have, by the SDP duality, that 
\begin{align}
    h_{\Rains}(\omega) = \sup_{\sigma \in \Rains(A:B)} \tr[\omega \sigma] = \inf_{\gamma \geq \omega} \|\gamma^{\sfT_B}\|_{\infty},
\end{align}
where $\|\cdot\|_\infty$ is the spectral norm. By the multiplicativity of $\|\cdot\|_\infty$ we can easily check that $h_{\Rains}$ is sub-multiplicative under tensor product, which is equivalent to the polar assumption (A.4) by Lemma~\ref{lema: polar set and support function}. The rest of assumptions can also be easily verified.

Therefore, we have the relaxation
\begin{align}
     D^\reg(\rho_{AB}\|\SEP) \geq D^\reg(\rho_{AB}\|\Rains),
\end{align}
where the right-hand side known as the regularized Rains bound can be efficiently estimated using Theorem~\ref{thm: efficient relative entropy program} by considering $\sA_n = \{\rho^{\ox n}\}$ and $\sB_n = \Rains(A^n:B^n)$ which has efficient SDP representations~\cite{fang2019non}.

\begin{remark}
    Similar to the regularized relative entropy of entanglement, the regularized Rains bound has the operational meaning~\cite{regula2019one} that it is the distillable entanglement under Rains-preserving operations $E_{D,\text{Rains}}$, that is,
\begin{align}
    E_{D,\text{Rains}}(\rho_{AB}) = D^\reg(\rho_{AB}\|\Rains).
\end{align}
    This marks the first time that an operational regularized entanglement measure has been shown to be efficiently computable, even when expressed as a regularized formula and beyond the zero-error setting. Previous work in~\cite{lami2024computable} studied the zero-error entanglement cost under PPT operations and proved that it is efficiently computable despite the absence of a closed-form formula.
\end{remark}

As the Rains-preserving operations is a superset of LOCC operations, the regularized Rains bound also gives an upper bound on the distillable entanglement under LOCC operations $E_{D,\text{\rm LOCC}}$, that is,
\begin{align}
    E_{D,\text{\rm LOCC}}(\rho_{AB}) \leq D^\reg(\rho_{AB}\|\Rains) \leq E_{D,\text{\rm ANE}}(\rho_{AB}).
\end{align}
This improves the best known efficiently computable bound for $E_{D,\text{\rm LOCC}}$ as well.

\subsection{Magic state distillation}
\label{sec: Magic state distillation}

The above argument for entanglement distillation also applies to magic states, which is a key resource for fault-tolerant quantum computing~\cite{Bravyi_2005,veitch2014resource,fang2020no}. The task of magic state distillation aims to extract as many copies of the target magic state as possible with asymptotically vanishing error. The distillable magic is denoted by $M_{D,\Omega}$, where $\Omega$ represents the set of allowed operations. Typically, the most natural choice of operations involves stabilizer operations, and the corresponding distillable magic is denoted by $M_{D,\text{STAB}}$. However, characterizing this set of operations is challenging.

Motivated by the idea of the Rains bound from entanglement theory, Wang et al.~\cite{Wang2018magicstates} relaxed the set of all stabilizer states to a set of sub-normalized states with non-positive mana, $\cW(\cH):=\{\sigma \in \PSD(\cH): \|\sigma\|_{W,1} \leq 1\}$, where $\|\cdot\|_{W,1}$ denotes the Wigner trace norm. Based on the set $\cW$, Wang et al. introduced a magic measure $D(\rho\|\cW)$, called thauma, and proved that it serves as an upper bound for the distillable magic under stabilizer operations, that is,
\begin{align}
    M_{D,\text{STAB}}(\rho) \leq D(\rho\|\cW) c(T),
\end{align}
where $c(T)$ is a constant that depends on the target magic state $T$.

Here, we can verify that $\cW$ satisfies our Assumption~\ref{ass: steins lemma assumptions} and apply Theorem~\ref{thm: efficient relative entropy program} to obtain an improved bound through regularization while keeping the computational efficiency.

To see this, it is straightforward to show that the support function of $\cW$ is given by
\begin{align}
    h_{\cW}(\omega) = \sup_{\sigma \in \cW} \tr[\omega \sigma] = \inf_{\gamma \geq \omega} \|\gamma\|_{W,\infty},
\end{align}
where the $\|\cdot\|_{W,\infty}$ is the Wigner spectral norm and the second equality follows from the SDP duality. Since $\|\cdot\|_{W,\infty}$ is multiplicative under tensor product, we can verify that the support function $h_{\cW}$ is sub-multiplicative under tensor product as well. Hence, the polar set $\cW^{\circ}$ is closed under tensor product by Lemma~\ref{lema: polar set and support function}, and the remaining assumptions in Assumption~\ref{ass: steins lemma assumptions} can also be verified. Then, we can consider the regularization and get
\begin{align}
    M_{D,\text{STAB}}(\rho) \leq D^\reg(\rho\|\cW) c(T),
\end{align}
where $D^\reg(\rho\|\cW)$ is the reguarlized thauma which remains efficiently computable by applying Theorem~\ref{thm: efficient relative entropy program} with $\sA_n = \{\rho^{\ox n}\}$ and $\sB_n = \cW(\cH^{\ox n})$. This improves the best-known estimation of magic state distillation under stablizer opeations.

\section{Conclusion}
\label{sec: conclusion}

We showed that regularized relative entropy between two sets of quantum states can be efficiently computed using convex optimization. This result has broad implications for quantum information theory, including the study of adversarial quantum channel discrimination, the estimation of the entanglement cost of quantum states and channels, entanglement distillation under LOCC operations and magic state distillation under stabilizer operations. Numerical experiments demonstrated that our new bounds outperform existing ones in various scenarios, even for the first level of approximation. Generally, our result can be applied by verifying the conditions of the relevant theory and performing necessary relaxations when required. Therefore, we anticipate that this approach has the potential for far-reaching applications beyond the specific cases discussed here.

Many problems remain open for future investigation. For example, while we have demonstrated that regularized relative entropies can be efficiently computed using convex optimization techniques, developing a more explicit algorithm and its implementation remains an area for further exploration. Additionally, designing a general algorithm to construct the smallest superset of a given set that satisfies the polar assumption presents an intriguing challenge. The solution of this would extend the applicability of our results to broader areas.

\vspace{1cm}

\noindent \textbf{Acknowledgements.} 
K.F. thanks Xin Wang for the updates on the $\PPT_k$ sets, and Ludovico Lami and Bartosz Regula for pointing out that $E_{\rm LR}$ vanishes for any full-rank state.
K.F. is supported by the National Natural Science Foundation of China (grant No. 92470113 and 12404569), the Shenzhen Science and Technology Program (grant No. JCYJ20240813113519025), the Shenzhen Fundamental Research Program (grant No. JCYJ20241202124023031), the 1+1+1 CUHK-CUHK(SZ)-GDST Joint Collaboration Fund (grant No. GRD\ P2025-022), and the University Development Fund (grant No. UDF01003565). H.F. was partially funded by UK Research and Innovation (UKRI) under the UK government’s Horizon Europe funding guarantee EP/X032051/1.
O.F. acknowledges support by the European Research Council (ERC Grant AlgoQIP, Agreement No. 851716), by the European Union’s Horizon 2020 research and innovation programme under Grant Agreement No 101017733 (VERIqTAS).

\bibliographystyle{alpha_abbrv}
\bibliography{Bib}

\end{document}